\documentclass[11pt]{article}
\usepackage{verbatim} %
\usepackage[usenames, dvipsnames]{color}

\usepackage{bbm}
\usepackage{mathtools}
\usepackage{amsmath}
\usepackage{amsthm, thmtools}
\usepackage{amstext,amssymb, amsfonts}
\usepackage{thm-restate}

\usepackage{dsfont} %
 
\usepackage{empheq} %

\usepackage{lmodern} %

\usepackage{float} %

\usepackage{array,multirow} %

\usepackage{algorithmic}
\usepackage[section]{algorithm} %
\usepackage{todonotes}

\newcounter{mynotes}
\usepackage{hyperref} %
\hypersetup{
    colorlinks=true,
    linkcolor=blue,
    filecolor=magenta,      
    urlcolor=cyan,
    citecolor=green!50!black,
}
 
\declaretheorem[within=section]{theorem}
\declaretheorem[sibling=theorem]{corollary}

\declaretheorem[sibling=theorem]{lemma}
\declaretheorem[sibling=theorem]{claim}

\declaretheorem[sibling=theorem]{definition}

\declaretheorem[sibling=theorem]{proposition}
\declaretheorem[sibling=theorem]{remark}
\declaretheorem[sibling=theorem]{question}
\declaretheorem[sibling=theorem]{conjecture}

\newcommand{\Q}{\mathbb{Q}} %
\newcommand{\Z}{\mathbb{Z}} %
\newcommand{\N}{\mathbb{N}} %

\newcommand{\F}{\mathbb{F}}
\newcommand{\K}{\mathbb{K}}
\newcommand{\Lb}{\mathbb{L}}
\newcommand{\PF}{\mathbb{P}\mathbb{F}} %

\newcommand{\cP}{\mathcal P}

\newcommand{\cS}{\mathcal S}

\newcommand{\supp}{\mathrm{supp}} %
\newcommand{\wt}{\mathrm{wt}} %

\newcommand{\inpro}[2]{\left\langle #1,#2 \right\rangle} %
\newcommand{\rk}{\mathrm{rank}} %

\newcommand{\poly}{\mathrm{poly}}

\renewcommand{\epsilon}{\varepsilon}
\newcommand{\eps}{\epsilon}

  \newcommand{\beq}{\begin{equation}}
  \newcommand{\eeq}{\end{equation}}
  \newcommand{\beqn}{\begin{equation*}}
  \newcommand{\eeqn}{\end{equation*}}
  \newcommand{\beqr}{\begin{eqnarray}}
  \newcommand{\eeqr}{\end{eqnarray}}
  \newcommand{\beqrn}{\begin{eqnarray*}}
  \newcommand{\eeqrn}{\end{eqnarray*}}
  \newcommand{\bmline}{\begin{multline}}
  \newcommand{\emline}{\end{multline}}
  \newcommand{\bmlinen}{\begin{multline*}}
  \newcommand{\emlinen}{\end{multline*}}

 \usepackage[margin=1in]{geometry}
\usepackage{url}

\renewcommand{\le}{\leqslant}
\renewcommand{\leq}{\leqslant}
\renewcommand{\ge}{\geqslant}
\renewcommand{\geq}{\geqslant}

\newcommand{\lp}{\left(}
\newcommand{\rp}{\right)}

\newcommand{\MDS}{\operatorname{MDS}}
\newcommand{\GZP}{\operatorname{GZP}}
\newcommand{\LDMDS}{\operatorname{LD-MDS}}
\newcommand{\Span}{\operatorname{span}}
\newcommand{\rank}{\operatorname{rank}}

\newcommand{\llangle}{\,\,(\!\!\!\!<\!}
\newcommand{\rrangle}{\!>\!\!\!\!)\,\,}
\newcommand{\crk}{{\rm crk}}
\newcommand{\ncrk}{{\rm ncrk}}

\newcommand{\change}[1]{#1}

\title{Generic Reed-Solomon Codes Achieve List-decoding Capacity}

\author{Joshua Brakensiek\thanks{University of California, Berkeley. Email: {\tt josh.brakensiek@berkeley.edu}. Portions of this work were completed
    at Stanford University and Microsoft Research, Redmond. Research
    supported in part by an NSF Graduate Research Fellowship and a
    Microsoft Research PhD Fellowship.} \and Sivakanth
  Gopi\thanks{Microsoft Research, Redmond, WA. Email: {\tt
      sigopi@microsoft.com}.} \and Visu Makam\thanks{Radix Trading
    Europe B. V. Email: {\tt visu@umich.edu}. Research supported by
    NSF Grant No. DMS-1638352, CCF-1412958, and CCF-1900460.}}

\date{}

\begin{document}

\maketitle

\begin{abstract}

In a recent paper, Brakensiek, Gopi and Makam~\cite{bgm2021mds} introduced \emph{higher order MDS codes} as a generalization of MDS codes. An order-$\ell$ MDS code, denoted by $\MDS(\ell)$, has the property that any $\ell$ subspaces formed from columns of its generator matrix intersect as minimally as possible. An independent work by Roth~\cite{roth2021higher} defined a different notion of higher order MDS codes as those achieving a generalized singleton bound for list-decoding. In this work, we show that these two notions of higher order MDS codes are (nearly) equivalent.

We also show that generic Reed-Solomon codes are $\MDS(\ell)$ for all $\ell$, relying crucially on the GM-MDS theorem which shows that generator matrices of generic Reed-Solomon codes achieve any possible zero pattern. As a corollary, this implies that generic Reed-Solomon codes achieve list decoding capacity. More concretely, we show that, with high probability, a random Reed-Solomon code of rate $R$ over an exponentially large field is list decodable from radius $1-R-\eps$ with list size at most $\frac{1-R-\eps}{\eps}$, resolving a conjecture of Shangguan and Tamo~\cite{shangguan2020combinatorial}.

\end{abstract}
\vspace{-3ex}
\thispagestyle{empty}
\newpage
\tableofcontents
\thispagestyle{empty}
\newpage
\section{Introduction}
The singleton bound states that a $\change{[n,k]}$-code can have distance at most $n-k+1$. Codes achieving this bound are called \emph{MDS codes}. Reed-Solomon codes \cite{reed1960polynomial} are an explicit construction of such codes over fields of size $O(n).$ In particular, they allow us to decode uniquely from up to half the minimum distance. List decoding was introduced independently by \cite{wozencraft1958list,elias1957list} to decode from beyond half the minimum distance. Naturally, we are not guaranteed to decode uniquely. But we can hopefully return a small list of codewords which are close to a corrupted codeword. We now define this formally.

\begin{definition} We say that a $\change{[n,k]}$-code $C$ is $(\rho,L)$-list decodable if there are at most $L$ codewords in any Hamming ball of radius $\rho n$.
\end{definition} We call $\rho$ the list-decoding radius and $L$ the list size.  A code with rate $R$ cannot be list decoded beyond radius $1-R$ with polynomial list size (see \cite{guruswami2012essential}). Therefore we must have $\rho\le 1-R$, this is called \emph{list-decoding capacity}. A code with rate $R$ which is $(1-R-\eps,L)$-list decodable for $L=L(\eps)$ is said to achieve list-decoding capacity. Here $\eps$ is called the \emph{gap to capacity}. It is known that random \change{\emph{non-linear}} codes can achieve list decoding capacity with list size $O(1/\eps)$ and alphabet size $\exp(1/\eps)$ (see~\cite{guruswami2012essential}). It is also known that random \emph{linear} codes over large enough alphabet achieve list-decoding capacity \cite{ZyablovP81}. There is also a stronger form of list decoding called \emph{average-radius list-decoding}.

\begin{definition} We say that a $\change{[n,k]}$-code $C$ is $(\rho,L)$-average-radius list-decodable\footnote{In some previous works, such as \cite{roth2021higher} this is referred to as \emph{strongly list-decodable}.} if there \change{do not} exist $L+1$ distinct codewords $c_0,c_1,\dots,c_L\in C$ and $y\in \F^n$, such that\footnote{Here $\wt(x)$ is the Hamming weight of $x$, i.e., the number of non-zero coordinates of $x.$} $$\frac{1}{L+1}\sum_{i=0}^L \wt(c_i-y) \le \rho n.$$
\end{definition} Note that if a code is $(\rho,L)$-average-radius list-decodable then it is also $(\rho,L)$-list decodable. The capacity for average case list-decoding is also $1-R$ and random linear codes achieve it \cite{ZyablovP81}.

A long line of work exists on constructing explicit codes which achieve list-decoding capacity. Following an initial breakthrough by~\cite{parvaresh2005correcting}, Folded Reed-Solomon codes of \cite{guruswami2008explicit} were the first explicit codes to achieve list-decoding capacity, but with list size and alphabet size polynomial in code length. Further works have reduced the list size to $\exp(\tilde{O}(1/\eps))$ where $\eps$ is the gap to capacity \cite{dvir2012subspace,guruswami2012folded,guruswami2013list,kopparty2018improved}. Some classes of structured random codes can also be shown to achieve list decoding capacity.  \cite{mosheiff2020ldpc} show that LDPC codes achieve list decoding capacity. \cite{guruswami2021punctured} show that puncturings of low-bias codes have good list-decodability.

\subsection{List-decoding Reed-Solomon Codes} Reed-Solomon codes are one of the most popular codes with several applications both in theory and practice \cite{wicker1999reed}.  We say that $C$ is a Reed-Solomon code if it has a generator matrix which is a \emph{Vandermonde} matrix. That is, there \change{exist} distinct $\alpha_1, \hdots, \alpha_n \in \F$ such that
\begin{equation}
\label{eq:Vandermonde_intro}
  \begin{pmatrix} 1 & 1 & \cdots & 1\\ \alpha_1 & \alpha_2 & \cdots & \alpha_n\\ \alpha_1^2 & \alpha_2^2 & \cdots & \alpha_n^2\\ \vdots & \vdots & \ddots & \vdots\\ \alpha_1^{k-1} & \alpha_2^{k-1} & \cdots & \alpha_n^{k-1}
  \end{pmatrix}
\end{equation} is a generator matrix of $C$. Reed-Solomon codes are MDS (maximum distance separable) codes, i.e., they achieve the maximum distance possible for a code with given rate. So naturally, there is a lot of interest in understanding the list-decodability of Reed-Solomon codes. %
\begin{question} Can Reed-Solomon codes achieve list-decoding capacity?
\end{question}
Guruswami and Sudan \cite{sudan1997decoding,Guruswami1998} showed that rate $R$ Reed-Solomon codes can be list-decoded from radius $1-\sqrt{R}$, which is also coincidentally the Johnson bound for list-decoding \cite{johnson1962new}. Whether Reed-Solomon codes can be list-decoded beyond the Johnson bound has been a topic of intense research. \cite{rudra2014every} are the first to show that random Reed-Solomon codes (i.e., when $\alpha_i$ are chosen randomly from a large enough field in (\ref{eq:Vandermonde_intro})) can be list-decoded beyond the Johnson bound in some parameter ranges. On the other hand, full length Reed-Solomon codes where $\alpha_1,\alpha_2,\dots,\alpha_n$ are all the field elements (with $n=|\F|$) are not list-decodable with constant list size with list-decoding radius $1-\alpha\sqrt{R}$ for sufficiently small constant $\alpha,R<1$, in fact the list size has to be at least $n^{2\log(1/\alpha)}$ \cite{ben2009subspace}. Therefore it is clear that the evaluation points $\alpha_i$ have to be chosen carefully from a large enough finite field to even beat the Johnson bound. In this paper, we are interested in understanding the list-decoding behavior of \change{\emph{generic} Reed-Solomon codes:}

\change{
\begin{definition}\label{def:grs}
Let $\F$ be an arbitrary finite field. For $n \in \N$, let $\mathbb K$ be the fractional field of the polynomial ring $\F[\alpha_1,\hdots, \alpha_n]$. A \emph{generic $[n,k]$-Reed-Solomon} code over $\F$ is the $k$-dimensional code $C_{gen}\subseteq\K^n$ generated by the generator matrix (\ref{eq:Vandermonde_intro}). Note that $\alpha_1,\dots,\alpha_n$ are symbolic variables with no relations between them.

Given an extension field $\Lb$ of $\F$ and a map $f : [n] \to \Lb$, the corresponding \emph{specialization} of $C_{gen}$ is the code $C \subseteq \Lb^n$ whose generator matrix is obtained by substituting $\alpha_i=f(i)$ in (\ref{eq:Vandermonde_intro}). A random Reed-Solomon code over $\Lb$ is obtained by choosing $f:[n]\to \Lb$ to be a uniformly random injective map.
\end{definition}
}

\change{If $\Lb$ is a sufficiently large field, the list-decoding behavior of a random specialization of a generic Reed-Solomon (also known as a \emph{randomly punctured} Reed-Solomon code) is the same as that of a generic Reed-Solomon code with high probability--see Proposition~\ref{prop:rRS_ldmds}. It also turns out that the list-decoding behavior of a generic Reed-Solomon code over $\Lb$ does not depend on $\F$ or its characteristic, which is not obvious \emph{a priori}--see Theorem~\ref{thm:main-dim}. Therefore, we are not concerned with the explicit choice of $\F$ in the rest of the paper, and often write $\F$ for $\Lb$.} 

In \cite{shangguan2020combinatorial}, Shangguan and Tamo made a startling conjecture that generic Reed-Solomon codes \change{do not} just beat the Johnson bound, but in fact achieve list-decoding capacity! They also conjectured a precise bound on the list size. In Section~\ref{sec:ld-mds}, we fully resolve their conjecture.

\begin{theorem}\label{thm:main-ld} Generic Reed-Solomon codes achieve list decoding capacity. If $C$ is a generic $\change{[n,k]}$-Reed-Solomon code with rate $R=k/n$, then $C$ is $(\rho,L)$-list decodable for
  \begin{equation}
    \label{eq:optimal_rho} \rho=1-R-\frac{1-R}{L+1}.
  \end{equation} Moreover, $C$ is also $(\rho,L)$-average-radius list-decodable for the same $\rho.$
\end{theorem} Equivalently, Theorem~\ref{thm:main-ld} shows that a generic Reed-Solomon code of rate $R$ is $(1-R-\eps,L)$-list decodable with $L=\frac{1-R-\eps}{\eps}.$ The bound (\ref{eq:optimal_rho}) is the best possible even for non-linear codes~\cite{goldberg2021list,roth2021higher}.

In Section~\ref{sec:ld-mds}, we also show how to turn Theorem~\ref{thm:main-ld} into a quantitative bound on the field size required for random Reed-Solomon codes to achieve list-decoding capacity.

\begin{theorem}
   \label{thm:main-ld-random} Let $n,k,L$ be positive integers and let $c(n,k,L) = 2Ln^2\binom{n}{\le n-k}^{L+1}$. A random $\change{[n,k]}$-Reed-Solomon code of rate $R=k/n$ over $\F$ is $(1-R-\frac{1-R}{L+1},L)$-average-radius list decodable with probability at least $1 - c / |\F|$.
 \end{theorem}

\begin{remark} Combining the construction from this paper with a different one in~\cite{bgm2021mds}, one can get $c(n,k,L) = n^{O(\min(k,n-k)L)}$ in Theorem~\ref{thm:main-ld-random} (see Remark~\ref{rem:bgm_construction}).
\end{remark}

\subsubsection{Previous Work} \cite{shangguan2020combinatorial} conjectured Theorem~\ref{thm:main-ld} and Theorem~\ref{thm:main-ld-random} and proved them in the case of $L=2,3.$\footnote{They did not conjecture average-radius list-decodability and also \change{did not} conjecture an explicit bound on $c(n,k,L).$} Note that it is also true for $L=1$ trivially, since Reed-Solomon codes are MDS. They also made an algebraic conjecture in their paper (see Conjecture 5.7 from \cite{shangguan2020combinatorial}) about the non-singularity of certain symbolic matrices, which would imply Theorem~\ref{thm:main-ld}. We prove this conjecture in Section~\ref{sec:conj_alg_st20}, the proof follows from some of the results in our paper which we use to prove Theorem~\ref{thm:main-ld}. Table~\ref{tab:RS_LD} shows prior results on list-decoding of random Reed-Solomon codes over fields of size $q.$

\bgroup \def\arraystretch{1.5}%
\begin{table}[h!]  \centering
\label{tab:RS_LD}
\begin{tabular}{ |c|c|c|c|c| } \hline & $\rho$ & $L$ & Rate $R$ & Field size $q$\\ \hline Johnson bound & $1-\eps$ & $qn^2$ & $\eps^2 $ & $n$\\ \hline \cite{rudra2014every} & $1-\eps$ & $O(1/\eps)$ & $\Omega \lp \frac{\eps}{\log(q)\log^5(1/\eps)}\rp$ & $\tilde{\Omega}(n/\eps)$\\ \hline \cite{shangguan2020combinatorial} & $1-R-\frac{1-R}{L+1}$& $L=2,3$ & $R$ & $\exp(n)$\\ \hline \cite{guo2020improved} & $1-\eps$ & $O(1/\eps)$ & $\Omega\lp\frac{\eps}{\log(1/\eps)}\rp$ & $(1/\eps)^{n}$\\ \hline \cite{ferber2022list} & $1-\eps$ & $\lceil 3/\eps \rceil$ & $\frac{\eps}{3(1+\zeta)}$ & $n^{1+1/\zeta}$ \\ \hline \cite{goldberg2021list} & $1-\eps$ & $O(1/\zeta)$ & $\frac{\eps-\zeta}{2-\eps+\zeta}$ & $\poly(n)$ \\ \hline Our work & $1-R-\eps$ & $\frac{1-R-\eps}{\eps}$ & $R$ & $\exp(\tilde{O}(n/\eps))$\\ \hline
\end{tabular}
 \caption{Adapted from \cite{goldberg2021list}. Prior works on list-decoding of random Reed-Solomon codes over fields of size $q.$}
\end{table} \egroup

\subsection{Higher Order MDS Codes}

Our results on list-decodability of generic Reed-Solomon codes follow from studying generalizations of MDS codes called \emph{higher order MDS codes}. We will show that generic Reed-Solomon codes are not just MDS, they are in fact higher order MDS codes. As we will see shortly, this implies that generic Reed-Solomon codes have optimal list-decodability. We will now dive into the rich theory of higher order MDS codes.

A $\change{[n,k]}$-code $C$ is MDS if it has the property that every non-zero codeword has hamming weight at least $n-k+1$. MDS codes have a number of equivalent characterizations. As has recently been explored in the literature~\cite{bgm2021mds,roth2021higher}, for many of the characterizations one can define a suitable generalization of MDS codes, deriving various notions of \emph{higher-order MDS codes}. Each of these has an order parameter $\ell \ge 1$, indicating the degree of generality over MDS codes.

\subparagraph{$\blacktriangleright \MDS(\ell)$.} Suppose $\change{G \in \F^{k\times n}}$ is the generator matrix of an $\change{[n,k]}$-code $C$ over $\F$. For $A \subset [n]$, let \change{$G|_A$ denote the submatrix of $G$ spanned by the columns of $G$ indexed by $A$. Further, let} $G_A$ denote the linear subspace of $\F^k$ spanned by \change{$G|_A$.} $C$ is MDS iff every $k$ columns of $G$ are linearly independent, equivalently $\dim(G_A)=\min\{|A|,k\}$ for all $A\subset [n].$ Equivalently, we can write this as $\dim(G_A)=\dim(W_A)$ where \change{$W$} is a generic \change{$k \times n$} matrix.\footnote{\change{A generic matrix is defined as in Definition~\ref{def:grs}. We consider a quotient field $\K$ of a polynomial ring with $kn$ variables, and have each entry of $W$ be a separate one of these variables.}} If $A,B\subset [n]$ are any two subsets, then
  \begin{align*} \dim(G_A \cap G_B)&=\dim(G_A)+\dim(G_B)-\dim(G_A+G_B)\\ &=\dim(G_A)+\dim(G_B)-\dim(G_{A\cup B})\\ &=\dim(W_A)+\dim(W_B)-\dim(W_{A\cup B})\\ &=\dim(W_A \cap W_B).
  \end{align*} Unfortunately, it may not be true that $\dim(G_{A_1} \cap G_{A_2} \cap G_{A_3})=\dim(W_{A_1} \cap W_{A_2} \cap W_{A_3})$ for all subsets $A_1,A_2,A_3\subset [n]$ if $C$ is MDS. This is because, the usual inclusion-exclusion principle fails for 3 or more subspaces. \cite{bgm2021mds} considered the following generalization of MDS codes which they called \emph{higher order MDS codes}.

\begin{definition}[$\MDS(\ell)$~\cite{bgm2021mds}] Let $C$ be an $\change{[n,k]}$-code with generator matrix $G$.  Let $\ell$ be a positive integer. We say that $C$ is $\MDS(\ell)$ if for any $\ell$ subsets $A_1, \hdots, A_\ell \subseteq [n]$ of size of at most $k$, we have that
  \begin{align} \dim (G_{A_1} \cap \cdots \cap G_{A_\ell}) = \dim (W_{A_1} \cap \cdots \cap W_{A_\ell}), \label{eq:MDS-L}
  \end{align} where \change{$W$ is a generic $k\times n$ matrix} over the same field characteristic.\footnote{Note that $\MDS(\ell)$ is a property of the code $C$ and not a particular generator matrix $G$ used to generate $C$. This is because if $G$ satisfies (\ref{eq:MDS-L}) then $MG$ also satisfies (\ref{eq:MDS-L}) for any $k\times k$ invertible matrix $M$.}
\end{definition}

Since $\dim(W_{A_1}\cap \dots \cap W_{A_\ell})$ is minimized when $W$ is a generic matrix, another intuitive way to think of an $\MDS(\ell)$ code is that $G_{A_1},G_{A_2},\dots,G_{A_\ell}$ intersect as minimally as possible for any $\ell$ subsets $A_1,A_2,\dots,A_\ell.$ The usual MDS codes are $\MDS(\ell)$ for $\ell=1,2$ by the above discussion.  This definition arose out of attempting to understand the properties of \emph{maximally recoverable tensor codes}, which are explained in more detail in Section~\ref{subsec:mrtc}. Briefly, the tensor product of $C$ and a parity check code is a maximally recoverable tensor code iff $C$ is a higher order MDS code of appropriate order (Proposition~\ref{prop:mrtc_mdsell}). Unlike MDS property, $\MDS(\ell)$ is not preserved under duality. The dual of an $\MDS(\ell)$ code is $\MDS(\ell)$ for $\ell\le 3$, but this fails for $\ell\ge 4$ \cite{bgm2021mds}.

To get some intuition for $\MDS(\ell)$, let's understand a $\change{[n,3]}$-code $C$ which is $\MDS(3).$ Let $v_1,v_2,\dots,v_n\in \F^3$ be the columns of a generator matrix of $C$. Since scaling the columns \change{does not} affect $\MDS(3)$, we can think of them as points in the projective plane $\PF^2.$ It is easy to see that $C$ is MDS iff the points $v_1,v_2,\dots,v_n \in \PF^2$ are in general position, that is no three points are collinear. $C$ is $\MDS(3)$ iff in addition, any 3 lines formed by joining disjoint pairs of points in $v_1,v_2,\dots,v_n$ are not concurrent.

\subparagraph{$\blacktriangleright \LDMDS(\ell)$.}  A generalization of the singleton bound was recently proved for list-decoding in \cite{shangguan2020combinatorial,roth2021higher,goldberg2021singleton}. Roth~\cite{roth2021higher} defined a higher order generalization of MDS codes as codes achieving this generalized singleton bound for list-decoding.\footnote{\cite{roth2021higher} also called these \emph{higher order MDS codes} independently of the prior work \cite{bgm2021mds}, leading to some confusion. Fortunately, as we will see shortly, these two notions are nearly equivalent.}

\begin{definition}[$\LDMDS(L)$~\cite{roth2021higher}] Let $C$ be a $\change{[n,k]}$-code. We say that $C$ is list decodable-$\MDS(L)$, denoted by $\LDMDS(L)$, if $C$ is $(\rho,L)$-average-radius list-decodable for $\rho=\frac{L}{L+1}\lp 1-\frac{k}{n}\rp.$ In other words, for any $y\in \F^n$, there \change{do not} exist $L+1$ \emph{distinct} codewords $c_0,c_1,\dots,c_L \in C$ such that
	\begin{equation}
		\label{eq:LDMDS_primal} \sum_{i=0}^L \wt(c_i-y) \le L(n-k).
	\end{equation} We say\footnote{In general, the notion of $\LDMDS(\ell)$ is not monotone in $\ell$.} that $C$ is $\LDMDS(\le L)$ if it is $\LDMDS(\ell)$ for all $1\le \ell \le L.$
\end{definition}

The list-decoding guarantees of $\LDMDS(L)$ are very strong. In particular, $\LDMDS(L)$ codes of rate $R$ get $\eps$-close to list-decoding capacity when $L\ge \frac{1-R-\eps}{\eps}$. Note that the usual MDS codes are $\LDMDS(1).$ \cite{roth2021higher} showed that $\LDMDS(L)$ property is preserved under duality only for $L=1,2$, and also gave some explicit constructions of $\LDMDS(2)$ codes.

\subparagraph{$\blacktriangleright \GZP(\ell)$.} In many coding theory applications, it is useful to have MDS codes with generator matrices having constrained supports, see \cite{dau2014gmmds,Halbawi2014distributed,yan2013algorithms,dau2013balanced} for some such applications to multiple access networks and secure data exchange. Dau et al. \cite{dau2014gmmds} have made a remarkable conjecture that Reed-Solomon codes over fields of size $q\ge n+k-1$ can have generator matrices with arbitrary patterns of zeros, as long as the pattern of zeros do not obviously preclude MDS property by having a large block of zeros. This came to be called the \emph{GM-MDS conjecture}. It was eventually proved independently by Lovett~\cite{lovett2018gmmds} and Yildiz and Hassibi~\cite{yildiz2019gmmds}. Before we state the GM-MDS theorem, we will make some crucial definitions.

Let $\cS=(S_1,S_2,\dots,S_k)$ where $S_1, \hdots, S_{k} \subset [n]$, we call such an $\cS$ a \emph{zero pattern} for $k\times n$ matrices. We say that $\cS$ has \emph{order} $\ell$ if there are $\ell$ distinct non-empty sets among $S_1,S_2,\dots,S_k.$ We say that a matrix $\change{G \in \F^{k \times n}}$ \emph{attains} the zero pattern $\cS$ if there exists an invertible matrix $\change{M \in \F^{k \times k}}$ such that $\widetilde{G} := M G$ has zeros in $\bigcup_{i=1}^{k} \{i\} \times S_i$. Note that $\widetilde{G}$ and $G$ generate the same code.  We now define the crucial notion of a \emph{generic zero pattern}.

\begin{definition}[Generic zero pattern] Suppose $\cS=(S_1,S_2,\dots,S_k)$ is a zero pattern for $k\times n$ matrices. We say $\cS$ is a \emph{generic zero pattern} if for all $I\subset [k]$,
  \begin{equation}
    \label{eq:gzp_hall} \left| \bigcap_{i\in I}S_i\right| \le k-|I|.
  \end{equation}
\end{definition}

It is not hard to see that, by Hall's matching theorem, (\ref{eq:gzp_hall}) is equivalent to the condition that a generic $k\times n$ matrix $W$ which has zeros in $\bigcup_{i=1}^{k} \{i\} \times S_i$ (and the rest of the entries of $W$ are generic), has all $k\times k$ minors non-zero. This is because the condition (\ref{eq:gzp_hall}) ensures that any $k\times k$ submatrix of $W$ has a matching of non-zero entries, and thus ensures non-zero determinant for this $k\times k$ submatrix. (\ref{eq:gzp_hall}) appeared in \cite{dau2014gmmds}, where it is called the MDS condition because it is a necessary condition for a $k\times n$ matrix with zeros in $\cS$ to be MDS. We will now define a new generalization of MDS codes, which we call $\GZP(\ell)$.

\begin{definition}[$\GZP(\ell)$] We say that a $\change{[n,k]}$-code $C$ is $\GZP(\ell)$ if $C$ is MDS \footnote{We explicitly add the MDS condition to avoid degenerate cases like zero matrix being $\GZP(\ell).$} and its generator matrix $\change{G \in \F^{k \times n}}$ attains all $k\times n$ generic zero patterns of order at most $\ell.$ \footnote{Note that $\GZP(\ell)$ is a property of the code $C$, and not a particular generator matrix $G$ used to generate $C$. This if because if $G$ attains a generic zero pattern, then $MG$ also attains it for any invertible $k\times k$ matrix $M.$}
\end{definition}

Thus, a code $C$ is $\GZP(\ell)$ if we can choose a generator matrix of $C$ to have any order $\ell$ generic zero pattern. One can prove that $\GZP(1)$ and $\GZP(2)$ are equivalent to MDS property, therefore this is indeed a generalization of the MDS property. Given how we defined $\GZP(\ell)$, it is not obvious to see why generic matrices should be $\GZP(\ell).$ In Proposition~\ref{prop:our-gm-mds}, we give an elementary proof of the fact that generic matrices are indeed $\GZP(\ell)$, i.e., a fixed generic matrix can attain any generic zero pattern.

 We will now state the GM-MDS theorem in terms of $\GZP(\ell)$ property.

\begin{theorem}[GM-MDS~\cite{dau2014gmmds,lovett2018gmmds,yildiz2019gmmds}]
\label{thm:gmmds_gzp} A generic Reed-Solomon code is $\GZP(\ell)$ for all $\ell$, i.e., a generic Reed-Solomon code can attain any generic zero pattern.
\end{theorem} The actual GM-MDS theorem says that for any particular generic zero pattern $\cS$, there exists a Reed-Solomon code over any field of size $q\ge n+k-1$ which can attain $\cS$. This is a simple consequence of Theorem~\ref{thm:gmmds_gzp}, but we are more interested in $\GZP(\ell)$ codes which simultaneously attain \emph{all} order $\ell$ generic zero patterns.

\subparagraph{Equivalence of \change{H}igher-order MDS \change{C}odes.} We are now ready to present the most important theorem of our paper. We show that, surprisingly, all these notions of higher-order MDS codes are equivalent (up to duality).%

\begin{restatable}{theorem}{mdsequiv}\label{thm:main-mds} The following are equivalent for a linear code $C$ for all $\ell\ge 1$.
  \begin{enumerate}
    \item[(a)] $C$ is $\MDS(\ell+1).$
    \item[(b)] $C^{\perp}$ is $\LDMDS(\le \ell).$
    \item[(c)] $C$ is $\GZP(\ell+1).$
  \end{enumerate}
\end{restatable}
\begin{proof} (a) iff (b) is proved in Section~\ref{sec:ld-mds} and (a) iff (c) is proved in Section~\ref{sec:gzp-mds}.
\end{proof}

\begin{remark} One can show that $\MDS(1),\MDS(2),\LDMDS(1),\GZP(1),\GZP(2)$ are all equivalent to MDS (see \cite{bgm2021mds,roth2021higher}).
\end{remark}

As we will see, the core of the proof of Theorem~\ref{thm:main-mds} is combinatorial, with some simple linear algebra on top. Since by GM-MDS theorem, generic Reed-Solomon codes are $\GZP(\ell)$ for all $\ell$, we have the following corollary.

\begin{corollary}
  \label{cor:main-RS} Generic Reed-Solomon codes are $\GZP(\ell),\MDS(\ell)$ and $\LDMDS(\ell)$ for all $\ell.$
\end{corollary}
\begin{proof} \change{Let $R \in \K^n$ be a generic Reed-Solomon code and let $G \in \K^{k \times n}$ its generator matrix and $H \in \K^{k \times n}$ be its parity check matrix. By the GM-MDS theorem (Theorem~\ref{thm:gmmds_gzp}), we have that $R$ is $\GZP(\ell)$. We prove in Proposition~\ref{prop:genericRS_duality} that, up to a scaling of the columns, $H$ is the generator matrix of a generic Reed-Solomon code. By inspection, each of the properties $\GZP(\ell),\MDS(\ell)$ and $\LDMDS(\ell)$ is invariant with respect to column scaling.} Therefore, \change{by Theorem~\ref{thm:main-mds}, we have that} generic Reed-Solomon codes are $\MDS(\ell)$ and $\LDMDS(\ell)$ for all $\ell$.
\end{proof} This immediately implies our main result that generic Reed-Solomon codes achieve list-decoding capacity (Theorem~\ref{thm:main-ld}).

\subsection{Proof Overview}

The proof of Theorem~\ref{thm:main-mds} has a few key steps, including proving some novel properties of generic zero patterns as well as solving a \emph{generic intersection problem}.

\paragraph{Dimension of generic intersections}

The core of the proof of Theorem~\ref{thm:main-mds} is a combinatorial characterization of the RHS of (\ref{eq:MDS-L}).

\begin{restatable}[Dimension of Generic Intersection]{theorem}{dimgen}\label{thm:main-dim} Given $A_1, \hdots, A_{\ell} \subseteq [n]$ of size at most $k$, for a \change{$k \times n$ generic matrix $W$}, we have that
\begin{align} \dim (W_{A_1} \cap \cdots \cap W_{A_{\ell}}) = \max_{P_1\sqcup P_2 \sqcup \dots \sqcup P_s=[\ell]} \left(\sum_{i\in [s]} \left|\bigcap_{j \in P_i} A_j\right| - (s-1) k\right) \label{eq:gen-dim}
\end{align} where the maximum is over all partitions of $[\ell].$ Note that the result is independent of the characteristic of the underlying field.
\end{restatable}

The proof of Theorem~\ref{thm:main-dim} appears in Section~\ref{sec:gzp_implies_mds}. Since the RHS of (\ref{eq:gen-dim}) has a maximum over exponentially many terms in $\ell$, it gives an $\lp\exp(\tilde{O}(\ell))k\rp$-time algorithm for computing the generic intersection dimension. In Section~\ref{app:compute}, we give a $\change{\poly(n,\ell)}$-time time algorithm for computing generic intersection dimension via LP duality and submodular optimization. In \change{Appendix}~\ref{sec:inv-theory}, we give an alternative $\change{\poly(n,\ell)}$-time algorithm to compute the RHS of (\ref{eq:gen-dim}) by reducing \emph{to a} non-commutative rank computation (see Theorem~\ref{thm:poly-time}).

In literature, one can find many problems that are similar or related to Theorem~\ref{thm:main-dim} in a range of subjects like Schubert calculus, intersection theory, matroid theory, representation stability and homological algebra to name a few. For example, it seems conceivable that there is a matroid-theoretic description or that there is a formula for the intersection dimension coming from Schubert calculus and intersection theory. Despite that, it seems very difficult to adapt the techniques from any of those subjects to say anything meaningful about the problem above, but a more in-depth analysis from the view-point of any of those subjects could lead to new insights in broader contexts (see Section~\ref{subsec:future}).

\paragraph{A novel characterization of sets in order-$\ell$ generic zero patterns} We also show a novel structural result on order $\ell$ generic zero patterns. In particular, if $\cS$ is an order $\ell$ generic zero pattern containing sets $A_1,A_2,\dots,A_\ell$ and say $d$ copies of the empty set. Then by applying (\ref{eq:gzp_hall}), one can easily show that for all partitions $\cP = P_1 \sqcup P_2 \sqcup \cdots \sqcup P_s = [\ell]$ we have that
\begin{equation}
\label{eq:partition_intro} \sum_{i=1}^s \left|\bigcap_{j \in P_i} A_j\right| \le (s-1)k+d.
\end{equation} Surprisingly the converse is also true (see Lemma \ref{lem:hall-part}). $A_1,A_2,\dots,A_\ell$ can be used to form an order $\ell$ generic zero pattern with $d$ copies of the empty set iff (\ref{eq:partition_intro}) holds. The proof involves an intricate induction, which on a high level is comparable to the induction used to prove Hall's matching theorem. In the proof, one identifies the partition $\cP$ for which the above inequality is tight (if no such partition exists, then one pads with elements). One can then recursively apply the induction hypothesis to each portion of the partition, which can then be combined together to show the result. This result is crucially used to prove that $\GZP(\ell)$ codes are also $\MDS(\ell)$ and to prove Theorem~\ref{thm:main-dim}.

The proof of Theorem~\ref{thm:main-dim} proceeds as follows. Let $d$ be the RHS of (\ref{eq:gen-dim}). By Lemma~\ref{lem:hall-part}, there exists a order $\ell$ zero pattern $\cS$ with copies of $A_1, \hdots, A_{\ell} \subseteq [n]$ and $d$ copies of the empty set. Then, since a generic matrix $W$ is $\GZP(\ell)$ (Proposition~\ref{prop:our-gm-mds}), there is an invertible $\change{M \in \K^{k \times k}}$ such that $\widetilde{W}=MW$ has the zero pattern $\cS$. From this, it is straightforward to upper bound the dimension of the intersection $\dim(W_{A_1}\cap \dots \cap W_{A_\ell})=\dim(\widetilde{W}_{A_1}\cap \dots \cap \widetilde{W}_{A_\ell})\le d.$ A matching lower bound follows from the pigeonhole principle and dimension counting. Note that this proof also immediately implies that $\GZP(\ell)$ codes are $\MDS(\ell)$ because in the proof of Theorem~\ref{thm:main-dim}, we only used the $\GZP(\ell)$ property of generic matrices to get the correct dimension.

\paragraph{A generalized Hall's theorem.} One of the key results in \cite{dau2014gmmds} is a \emph{generalized Hall's theorem}. For any generic zero pattern $\cS = (S_1, \hdots, S_k)$ there exists a generic zero pattern $\cS' = (S'_1, \hdots, S'_k)$ which contains $\cS$ (i.e., for all $i$, $S'_i \supseteq S_i$) such that $|S'_i| = k-1$ for all $i$. However, this theorem does not preserve order, if $\cS$ is order $\ell$, the order of $\cS'$ can be as large as $k$ (and in fact must equal $k$).

In Section~\ref{subsec:gen-hall}, we further generalize the generalized Hall's theorem from \cite{dau2014gmmds}. In particular, we show that if $\cS$ is an order $\ell$ generic zero pattern, then there is an order $\ell$ generic zero pattern $\cS'$ which contains $\cS$ such that for the non-empty $A_1, \hdots, A_\ell$ which define $\cS'$, each $A_i$ appears exactly $k - |A_i|$ times. Such an $\cS'$ is called \emph{maximal}.

This new generalized Hall's theorem is used to prove that $\MDS(\ell)$ codes are also $\GZP(\ell)$. Suppose $C$ is an $\MDS(\ell)$ code with generator matrix $G$. Given an order $\ell$ generic zero pattern $\cS$, one uses our generalized Hall theorem (Theorem~\ref{thm:ell-hall}), to find a maximal order $\ell$ generic zero pattern $\cS'$ containing $\cS$. Let $A_1, \hdots, A_\ell$ be the $\ell$ non-empty sets in $\cS'$ and say $\cS'$ has $d$ copies of the empty set. Using the $\MDS(\ell)$ property of $G$ and the fact that $\cS'$ satisfies (\ref{eq:gzp_hall}), we can prove that $\dim(G_{A_1} \cap \cdots \cap G_{A_\ell})=d$ via Theorem~\ref{thm:main-dim}. By taking the dual of this intersection and performing a dimension-counting argument, one can show that\footnote{The dual is of the linear space $G_{A_1} \subseteq \F^k$. We are \emph{not} taking the dual of the original matrix $G$.}  $(G_{A_1})^{\perp}, \hdots, (G_{A_{\ell}})^{\perp}$ are linearly independent. One can then show that bases for these spaces can be put together to build a matrix $M$ such that $MG$ has the desired zero pattern. For the proof to work, we absolutely need the fact that each $A_i$ appears exactly $k-|A_i|$ times in the pattern $\cS'$, which is guaranteed by the generalized Hall's theorem.

\paragraph{Equivalence of $\MDS(\ell+1)$ and $\LDMDS(\le \ell)^{\perp}$.}  The proof of equivalence mostly follows from Theorem~\ref{thm:main-dim}.  We prove the contrapositive: that $C$ is \emph{not} $\MDS(\ell+1)$ iff $C^{\perp}$ is \emph{not} $\LDMDS(\le \ell)$.

Let $G$ be a generator matrix of $C$, note that $G$ is a parity check matrix for $C^\perp$. If $C$ is \emph{not} $\MDS(\ell+1)$, there exists some choice of $A_1, \hdots, A_{\ell+1}$ for which (\ref{eq:gen-dim}) is not satisfied. In fact, using a result of \cite{bgm2021mds}, one can assume that the RHS of (\ref{eq:gen-dim}) is $0$. This implies there is a nontrivial $z \in G_{A_1} \cap \cdots \cap G_{A_{\ell+1}}$ which is not captured by a generic intersection. In particular, for all $i \in [\ell]$, there is $u_i \in \F^n$ with $\supp(u_i) \subseteq A_i$ \change{(where for $u \in \F^n$, $\supp(u) := \{i \in [n] : u_i \neq 0\}$)} with $z = G u_i$. This is almost enough to prove that $C^\perp$ is not $\LDMDS(\ell)$, but some of the $u_i$'s may be equal. To get around this, we consider a partition of $[\ell]$ with two $A_i$'s in the same part if their $u_i$'s are equal. The resulting inequality arising from using this partition with (\ref{eq:gen-dim}) is enough to prove that $C^\perp$ is not $\LDMDS(\ell')$ for some $\ell' \le \ell$.

Now assume that $C^\perp$ is not $\LDMDS(\le \ell)$, WLOG say $C^\perp$ is not $\LDMDS(\ell)$. In particular, this implies that there are \emph{distinct} $u_1, \hdots, u_{\ell+1}$ for which $Gu_1 = \cdots = Gu_{\ell+1}$ and $\sum_{i=1}^{\ell+1} \wt(u_i) \le \ell k.$ One can then let $A_i = \supp(u_i)$, and consider the intersection $G_{A_1} \cap \cdots \cap G_{A_{\ell+1}}$. Using the distinctness of $u_1,\dots,u_{\ell+1},$ one can prove that the dimension of this intersection is strictly greater than the corresponding generic intersection. This is enough to show that $C^{\perp}$ is not $\MDS(\ell)$.

\subsection{Further Applications and Connections} In this section, we mention some further applications and connections of our work to different areas of coding theory and mathematics.

\subsubsection{Generic Gabidulin \change{C}odes \change{A}chieve \change{L}ist-decoding \change{C}apacity} Let $\alpha_1,\alpha_2,\dots,\alpha_n$ be linearly independent over some base field $\F_q.$ A Gabidulin code has the following generator matrix:
\begin{equation}
\label{eq:Gabidulin_intro}
  \begin{pmatrix} \alpha_1 & \alpha_2 & \cdots & \alpha_n\\ \alpha_1^q & \alpha_2^q & \cdots & \alpha_n^q\\ \alpha_1^{q^2} & \alpha_2^{q^2} & \cdots & \alpha_n^{q^2}\\ \vdots & \vdots & \ddots & \vdots\\ \alpha_1^{q^{k-1}} & \alpha_2^{q^{k-1}} & \cdots & \alpha_n^{q^{k-1}}
  \end{pmatrix}.
\end{equation} A generic Gabidulin code is defined \change{as in Definition~\ref{def:grs}} by choosing $\alpha_1,\alpha_2,\dots,\alpha_n$ \change{as generators of the quotient field of a polynomial ring.} Gabidulin codes are rank metric codes which achieve the rank metric singleton bound with applications in network coding, space-time coding and cryptography~\cite{gabidulinrank,gabidulin2021rank}.  GM-MDS theorem was extended to Gabidulin codes over both finite and zero characteristic in \cite{yildiz2019gabidulin,yildiz2020gabidulincharzero}.  Thus generator of matrices of generic Gabidulin codes over both finite and zero characteristic also satisfy $\GZP(\ell)$ for all $\ell.$ Since dual of a generic Gabidulin code is also a generic Gabidulin code \cite{gabidulin2021rank}, Theorem~\ref{thm:main-mds} shows that Gabidulin codes are $\MDS(\ell)$ and $\LDMDS(\ell)$ for all $\ell.$ This implies that generic Gabidulin codes have optimal list-decoding guarantees with respect to the Hamming metric.
\begin{theorem}\label{thm:gabidulin} Generic Gabidulin codes achieve list-decoding capacity (in the Hamming metric). In particular, they are $(\rho,L)$-average-radius list-decodable for $$\rho=1-R-\frac{1-R}{L+1}.$$
\end{theorem}

\subsubsection{Field \change{S}ize \change{L}ower \change{B}ounds for \texorpdfstring{$\LDMDS(L)$}{LDMDS(L)} \change{C}odes} The field size lower bound on $\MDS(\ell)$ of~\cite{bgm2021mds} is as follows.

\begin{proposition}[Corollary~4.2~\cite{bgm2021mds}]
\label{prop:lb_mdsell} If $C$ is an $\change{[n,k]}$-code over $\F$ which is $\MDS(\ell)$, then
  \[ |\F| \gtrsim_\ell n^{\min\{k,n-k,\ell\}-1}.
  \]
\end{proposition}

As an immediate corollary of Theorem~\ref{thm:main-mds}, we have the following:

\begin{corollary}
\label{prop:lb_ldmds} If $C$ is an $\change{[n,k]}$-code over $\F$ which is $\LDMDS(\le L)$, then
  \[ |\F| \gtrsim_L n^{\min\{k-1,n-k-1,L\}}.
  \]
\end{corollary}

In particular, if $C$ is a code of constant rate (thus both $k$ and $n-k$ tend to infinity), we have that $|\F| \gtrsim_L n^L$. In this constant-rate regime, our lower bound significantly improves upon the lower bound from \cite{roth2021higher} which says that $|\F| \gtrsim_L \lp\frac{n}{n-k}\rp^{\min\{k-1,L\}}$.

\subsubsection{Maximally \change{R}ecoverable \change{T}ensor \change{C}odes}\label{subsec:mrtc}

Gopalan et al. \cite{Gopalan2016} introduced the notion of maximally recoverable (MR) codes with grid-like topologies. These codes have applications in distributed storage in datacenters, where they offer a good trade-off between low latency, durability and storage efficiency \cite{huang2012erasure}.  An important special case of such codes are MR tensor codes. A code $C$ is a \emph{$(m,n,a,b)$-tensor code} if it can be expressed as $C_{col} \otimes C_{row}$, where $C_{col}$ is a $\change{[m,m-a]}$ code and $C_{row}$ is a $\change{[n,n-b]}$ code. In other words, the codewords of $C$ are $m\times n$ matrices where each row belongs to $C_{row}$ and each column belongs to $C_{col}.$ There are $a$ parity checks per column and $b$ parity checks per row. For example, the f4 storage architecture of Facebook (now Meta) uses an $(m=3,n=14,a=1,b=4)$-tensor code~\cite{muralidhar2014f4}. Such a code $C$ is \emph{maximally recoverable} if it can recover from every erasure pattern $E \subseteq [m] \times [n]$ which can be recovered from by choosing a generic $C_{col}$ and $C_{row}$. Thus MR tensor codes are optimal codes since they can recover from any erasure pattern that is information theoretically possible to recover from. MR tensor codes are poorly understood with no known explicit constructions. Even a characterization of which erasure patterns are correctable by an $(m,n,a,b)$-MR tensor code is not known except in the case of $a=1$~\cite{Gopalan2016}.  \cite{bgm2021mds} defined $\MDS(\ell)$ codes motivated by the following proposition.
\begin{proposition}[\cite{bgm2021mds}]
\label{prop:mrtc_mdsell} Let $C=C_{col}\otimes C_{row}$ be an $(m,n,a=1,b)$-tensor code. Here $a = 1$ and thus $C_{col}$ is a parity check code. Then $C$ is maximally recoverable if and only if $C_{row}$ is $\MDS(m)$.
\end{proposition} Thus, better understanding higher order MDS codes is essential to understanding maximally recoverable tensor codes. We hope that Theorem~\ref{thm:main-mds} which shows the importance of higher order MDS codes to various areas of coding theory, will help in designing explicit $\MDS(\ell)$ codes and thus explicit maximally recoverable tensor codes. The following is a direct corollary of Proposition~\ref{prop:mrtc_mdsell} and Corollary~\ref{cor:main-RS}.

\begin{corollary} The tensor product of a parity check code and a generic Reed-Solomon code is maximally recoverable.
\end{corollary}

\subsubsection{Efficient Computation of Generic Intersection Dimension}

Although Theorem~\ref{thm:main-dim} produces a closed formula for generic intersection dimension, there are exponentially many (in \change{$\ell$}) partitions to consider. This can make computing the dimension cumbersome. We give \change{two} deterministic\footnote{We remark that a simple randomized polynomial-time algorithm to compute generic intersection dimension is to randomly sample $W$ over a large enough field and compute the intersection dimension directly by a rank computation \cite{bgm2021mds}.} polynomial-time algorithm\change{s for computing the generic intersection dimension.}

\begin{restatable}{theorem}{polytime}\label{thm:poly-time} Given $k, n$ and $A_1 \dots A_\ell \subseteq [n]$, we can compute the intersection dimension $\dim (W_{A_1} \cap \cdots \cap W_{A_{\ell}}) $ for generic $W$ in $\change{\poly(n,\ell)}$-time.
\end{restatable}
\begin{remark} Explicitly computing the formula found by Theorem~\ref{thm:main-dim} takes $\change{nk\exp(\tilde{O}(\ell))}$ time, as there are $\exp(\tilde{O}(\ell))$ partitions of $[\ell]$~(e.g., \cite{de1981asymptotic}). Thus, Theorem~\ref{thm:poly-time} is superior when $\ell \ge \operatorname{polylog}(\change{n})$.
\end{remark}

\change{The first algorithm, presented in Section~\ref{app:compute}, uses techniques in combinatorial optimization, including finding an integral solution to a linear-programming relaxation via submodular optimization. The second algorithm, presented in Appendix~\ref{sec:inv-theory}, utilizes techniques in \emph{invariant theory}. In particular, the algorithm reduces to computing the \emph{non-commutative rank} of a suitable linear matrix. The techniques used are based on recent breakthroughs in invariant theory \cite{GGOW,IQS18,IQS17,DM17} and} are part of a larger ambitious program of Mulmuley \cite{GCTV} that attempts to approach central problems in complexity via orbit problems in invariant theory that has seen much progress over the last decade, see \cite{BFGOWW} and references there-in.

\subsection{Subsequent Work and Open Questions}\label{subsec:future}

Since this paper was original posted, there have been a number of exciting developments in the theory of higher order MDS codes and related questions.

\begin{itemize}
\item The works of Guo-Zhang~\cite{guo2023randomly} and Alrabiah-Guruswami-Li~\cite{alrabiah2023randomly} showed that random Reed-Solomon codes achieve list-decoding capacity over fields of size $O_{\eps}(n^2)$ and $O_{\eps}(n)$, respectively. The latter paper also proves random linear codes achieve capacity over fields of size $\exp(O(1/\eps^2))$. They show that in order to achieve list-decoding capacity, it suffices to have only a sparse subset of the higher order MDS conditions to hold. With clever combinatorial arguments and the fact (by Corollary~\ref{cor:main-RS}) that all these determinants generically non-vanish, they show that this sparsity allows for one to have exponential gains in success probability when compared to the Schwartz-Zippel analysis in Section~\ref{sec:ld-mds}. \change{Since then, this result has also been extended to AG codes~\cite{brakensiek2023generalized,bdgz2023}, where they show that random AG codes of field size $\exp(O(1/\eps^2))$ achieve list-decoding capacity, matching the bound on random linear codes.}

\item The work of Brakensiek-Dhar-Gopi~\cite{brakensiek2023improved} prove that a $\change{[n,k]}$-$\MDS(3)$ code require field size at least $\binom{n-2}{k-1}+1$, which is exponential in $n$ in the constant-rate regime. This lower bound supersedes nearly all previously known lower bound for higher order MDS codes. This was later extended by Alrabiah-Guruswami-Li~\cite{alrabiah2023ag} to show that to get $\eps$-close to list-decoding capacity, the code (even nonlinear!) requires field size at least $\exp(\Omega(1/\eps))$. They show also show this bound is essentially tight for non-linear codes.

\change{
\item The connections between list-decoding and GM-MDS have spurred further generalizations of GM-MDS. The work of Brakensiek-Dhar-Gopi~\cite{brakensiek2023generalized} extended the GM-MDS theorem from Reed-Solomon codes to arbitrary polynomial codes. Further, Ron-Zewi, Venkitesh, and Wootters~\cite{ron2024efficient} prove a generalization of GM-MDS to ``polynomial ideal codes'' which captures families of codes like folded Reed-Solomon codes which are codes defined over an extension field, but they are only linear in the base field. Their work also gives efficient list-decoding algorithms for random polynomial ideal codes. However, efficient list decoding of random Reed-Solomon codes is still open (see Question~\ref{ques:rs}).

\item The work of Guo-Xing-Yuan-Zhang~\cite{guo2024random} has generalized the notion of higher order MDS codes, where the metric is Hamming distance, to rank metric codes. In this setting, they construct analogues of $\GZP(\ell), \MDS(\ell), and \LDMDS(\ell)$ called $\operatorname{GKP}(\ell),\operatorname{MRD}(\ell),\operatorname{LDMRD}(\ell),$ respectively, and show that the analogue of Theorem~\ref{thm:main-mds} holds. They also prove a generalized GM-MDS theorem that generic Gabidulin are $\operatorname{GKP}(\ell)$ for all $\ell \ge 2$; and thus generalize Theorem~\ref{thm:gabidulin} by showing that generic Gabidulin codes achieve list decoding capacity with respect to the rank metric.
}
\end{itemize}

Even with all this progress, there are still a number of exciting directions that warrant further exploration.

\subparagraph{Constructions of \change{H}igher-order MDS \change{C}odes.} Despite knowing that generic Reed-Solomon codes are higher-order MDS, we do not know of any good explicit constructions of such higher-order MDS codes in general.\footnote{Note that one can always get an ``explicit'' construction over doubly exponential size fields by choosing $\alpha_i$ to be in a degree $k$ field extension over $\F_2(\alpha_1,\dots,\alpha_{i-1})$ in (\ref{eq:Vandermonde_intro}) \cite{shangguan2020combinatorial}. A more sophisticated construction with better field size (although still doubly exponential) appears in \cite{brakensiek2023improved}.}

As previously mentioned, Theorem~\ref{thm:main-ld-random} and the results of \cite{bgm2021mds} imply that $\MDS(\ell)$ codes exist over fields of size $n^{O(\min\{k,n-k\}\ell)}$.

   \change{Conversely, $\MDS(\ell)$ codes require fields of size $\Omega_{\ell}(n^{\min\{\ell,k,n-k\}-1})$~\cite{bgm2021mds} (see Proposition~\ref{prop:lb_mdsell}). Later on, Brakensiek-Dhar-Gopi~\cite{brakensiek2023improved} showed that even $\MDS(3)$ codes require fields of size $\ge \binom{n-2}{k-1}+1$.}
     
   The simplest non-trivial case where we don't know the correct field size is when $k=3,\ell=3$. The lower bounds imply that, we need field size at least $q=\Omega(n^2)$.
\begin{question} \change{More} concretely, can we construct an explicit $\MDS(3)$ $\change{[n,3]}$-code over a field of size $O(n^2)$?
\end{question}

Note that \cite{bgm2021mds} gives $O(n^2)$ field size constructions for notions slightly weaker than $\MDS(3)$. \cite{roth2021higher} gives an explicit construction of size $O(n^{32})$ and a non-explicit construction over fields of size $O(n^5)$. This was improved to an explicit construction over fields of size $O(n^3)$ by ~\cite{brakensiek2023improved}. More generally~\cite{roth2021higher} gives an explicit construction of $\change{[n,k]}$-$\MDS(3)$ code over fields of size $O(n^{k^{2k}}).$

\subparagraph{Maximally \change{R}ecoverable \change{T}ensor \change{C}odes \change{W}hen $a,b\ge 2.$} We saw that $\MDS(\ell)$ codes arise naturally from studying $(m,n,a,b)$-MR tensor codes when $a=1$. It would be interesting to study, what properties of the row and column codes would be needed to construct MR tensor codes for $a,b\ge 2.$

   As previously mentioned, the work of~\cite{Gopalan2016} fully characterized the correctable erasure patterns for a $(m,n,a,b)$ tensor code when $a = 1$. Theorem~\ref{thm:main-dim}, when combined with the results of~\cite{bgm2021mds}, fully characterizes the linearly independent patterns when $a = 1$. We hope that these results lead to insights which resolve question of characterizing correctable erasure patterns in the general case. More precisely,

\begin{question} Given generic vectors $u_1, \hdots, u_m \in \F^{m-a}$ and $v_1, \hdots, v_n \in \F^{n-b}$, for which $E \subseteq [m] \times [n]$ is $\{u_i \otimes v_j : (i,j) \in E\}$ of full rank? For which $E$ are they linearly independent?
\end{question}

\change{The subsequent work of Brakensiek-Dhar-Gao-Gopi-Larson~\cite{brakensiek2024rigidity} has found a characterization in the case that $a = b = 2$. However, the computational complexity of determining whether $\{u_i \otimes v_j : (i,j) \in E\}$ has full rank lies in $\mathsf{NP}$.}

\subparagraph{Efficient \change{L}ist-decoding of $\LDMDS(\le L)$ \change{C}odes.} As previously mentioned, the result of Guruswami-Sudan~\cite{Guruswami1998} shows that any $\change{[n,k]}$-Reed-Solomon code of rate $R$ can be efficiently list-decoded up to radius $\rho = 1-\sqrt{R}$. A hardness result by Cheng and Wan~\cite{cheng2007list} states that it is discrete-logarithm-hard to decode up to radius $\hat{\rho} := 1 - \hat{g}/n$, where
\[ \hat{g} = \min \left\{g : \binom{n}{g}|\F|^{k-g} \le 1\right\}
\] However, this result only applies for small field sizes. Since we now know that $|\F| \ge \binom{n-2}{k-1}$, we have that $\hat{g} = k + \Theta(1)$, which is essentially the list-decoding capacity bound of radius $1-R$. On the other hand, $\LDMDS(\le L)$ are list decodable only up to list-decoding radius $\rho=1-R - \frac{1-R}{L+1}$ with list size $L$.  Thus, we believe the following is open in general.

\begin{question}\label{ques:rs} Assume $C$ is a $\change{[n,k]}$-Reed-Solomon code which is $\LDMDS(\le L)$. Given $y \in \F^n$, can one efficiently list all $c \in C$ with distance from $y$ at most $\frac{L}{L+1}(n-k)$?
\end{question}

\change{The best progress toward answering this question so far is by Ron-Zewi, Venkitesh, and Wootters~\cite{ron2024efficient}, who give efficient list-decoding algorithms for random polynomial ideal codes.}

\subsection*{Notation}
\label{sec:prelims}

A linear $\change{[n,k]}$-code $C$ is a $k$-dimensional subspace of $\F^n$.\footnote{In this paper, we will only work with linear codes. So unless specified otherwise, a $\change{[n,k]}$-code is always a linear code.} A matrix $\change{G \in \F^{k \times n}}$ is a generator matrix of $C$, if the rows of $G$ are a basis of $C$. A matrix $\change{H \in \F^{(n-k)\times n}}$ is called a parity check matrix for $C$ if $C=\{x:Hx=0\}$. The dual code $C^\perp$ is defined as $C^\perp=\{y:\inpro{x}{y}=0\ \forall x\in C\}.$ $C^\perp$ is a $\change{[n,n-k]}$-code and its generator matrix is the parity check matrix of $C.$

We let $[n]$ denote the set $\{1,2, \hdots, n\}$. Given a collection of sets $A_1, \hdots, A_k \subseteq [n]$, and a nonempty set $I \subseteq [k]$, we let $A_I = \bigcap_{i \in I} A_i$.

Let $V$ be a $k \times n$ matrix. For all $i \in [n]$, let $v_i$ denote the $i$th column of $V$. Given $A \subseteq [n],$ we let $V_A = \Span\{v_j : j \in A\}$. This notation should not be confused with the $A_I$ notation.

\subsection*{Organization}

The remainder of the paper is organized as follows. In Section~\ref{sec:gzp}, we discuss generic zero patterns in more detail, particularly how they relate to the GM-MDS conjecture. In Section~\ref{sec:gzp-mds}, we prove that $\GZP(\ell)$ and $\MDS(\ell)$ are equivalent. In the process, we prove Theorem~\ref{thm:main-dim}. In Section~\ref{sec:ld-mds}, we prove that $\MDS(\ell)$ and $\LDMDS(\le \ell-1)$ are equivalent, up to duality. This completes the proofs of Theorem~\ref{thm:main-ld} and Theorem~\ref{thm:main-mds}. We also prove Conjecture 5.7 from \cite{shangguan2020combinatorial}. In Section~\ref{app:compute}, we give a deterministic polynomial-time algorithm for computing the generic intersection dimension. In \change{Appendix}~\ref{sec:inv-theory}, we show that Theorem~\ref{thm:main-dim} also holds in the non-commutative setting, yielding an alternative deterministic polynomial-time algorithm for generic intersection dimension.

\subsection*{Acknowledgments}

We thank Venkatesan Guruswami, Sergey Yekhanin, and June Huh for valuable discussions and encouragement. \change{We thank Zihan Zhang for pointing out an error in the padding argument in the proof of Lemma~\ref{lem:hall-part}.} We thank anonymous reviewers for numerous helpful comments \change{and corrections}.

\section{Generic Zero Patterns (GZPs)}
\label{sec:gzp}

Recall that a zero pattern $\cS=(S_1,S_2,\dots,S_k)$ is called a generic zero pattern if
\begin{equation}
	\label{eqn:GZP_hall} \left|\bigcap_{i\in I} S_i\right|\le k-|I| \ \ \forall I\subset [k].
\end{equation} Also recall that (\ref{eqn:GZP_hall}) is equivalent to the fact that a generic matrix with zero pattern $\cS$ has all of its $k \times k$ minors non-zero.  We will now prove a generic matrix can attain any generic zero pattern.
\begin{proposition}\label{prop:our-gm-mds} A generic $k\times n$ matrix can attain any $k\times n$ generic zero pattern. In other words, generic codes are $\GZP(\ell)$ for all $\ell.$
\end{proposition}
\begin{proof} Let $\change{V \in \K^{k\times n}}$ be a generic matrix, which is the generator matrix of a generic $\change{[n,k]}$-code.  Let $\cS=(S_1,\dots, S_{k})$ be a generic zero pattern for $k\times n$ matrices, i.e. $\cS$ satisfies (\ref{eqn:GZP_hall}). We want to show that there exists some invertible matrix $\change{M \in \K^{k \times k}}$ such that $MV$ has zeros in $\cup_{i\in [k]} \{i\}\times S_i.$ By the generalized Hall's theorem (Theorem \ref{thm:gen-hall}), there exists $S_i'\subset [n]$ such that $S_i' \supset S_i$, $|S_i'|=k-1$ and $(S_1',S_2',\dots,S_k')$ satisfy (\ref{eqn:GZP_hall}). Therefore, WLOG we can assume that $|S_i|=k-1$ for all $i.$ Let $v_1,v_2,\dots,v_n$ be the columns of $V$. Let $m_1,m_2,\dots,m_k$ be the rows of $M$. For $MV$ to have zeros in $i\times {S_i}$, it must be that $\inpro{m_i}{v_j}=0$ for all $j\in S_i.$ Therefore $m_i = V_{S_i}^\perp$ (up to scaling), note that $V_{S_i}^\perp$ is a one-dimensional space since $|S_i|=k-1$. Moreover the entries of $m_i=V_{S_i}^\perp$ can be expressed as $(k-1)\times (k-1)$ minors of $V_{S_i}$ (with some $\pm$ signs) which are some polynomials in the entries of $V$. Therefore $M$ is completely determined (up to scaling of rows) by $V$, and the entries of $M$ are some polynomials in the entries of $V$. Now we just need to prove that $\det(M)$ which is a polynomial in the entries of $V$ is not identically zero.

To prove this, we give a particular setting of
$V=V^* \in \F^{k\times n}$ for which $\det(M)\ne 0$, for any large
enough field $\F$ (of any characteristic). Set all the entries
$V^*_{ij}=0$ whenever $j\in S_i$ and set the remaining entries
randomly from $\F$. Since $\cS$ is a generic zero pattern, $V^*$ is an
MDS matrix with high probability by Hall's matching theorem. Therefore
$m_i=(V^*_{S_i})^\perp=e_i$ (up to scaling), where $e_i$ is the $i$th
standard basis vector. Therefore $M=I_{k\times k}$ is the $k\times k$
identity matrix (upto scaling of rows), which has non-zero
determinant.
\end{proof}

\subsection{Generalized Hall's Theorem and Maximal GZPs}\label{subsec:gen-hall}

While formulating the GM-MDS conjecture, \cite{dau2014gmmds} proved a variant of a Generalized Hall's Theorem and used it to show that any generic zero pattern can be extended to a maximal generic zero pattern.

\begin{theorem}[Generalized Hall's Theorem--modern statement~\cite{dau2014gmmds}]\label{thm:gen-hall} Let $\cS = (S_1, \hdots, S_k)$ be a generic zero pattern for $\change{[n,k]}$-codes. Then, there exists a generic zero pattern $\cS' = (S'_1, \hdots , S'_k)$ such that for all $i \in [k]$, $|S'_i| = k-1$ and $S_i \subseteq S'_i$.
\end{theorem}

\begin{remark} Note that Theorem~\ref{thm:gen-hall} is a generalization of the classic Hall's theorem about existence of a bipartite matching when $k=n.$ In this case, if we form a bipartite graph between $[k]$ and $[n]$ where $i\in [k]$ has neighborhood $\bar{S_i}$, (\ref{eqn:GZP_hall}) becomes the Hall's matching condition that the neighborhood $N(I)$ of any set $I$ satisfies $|N(I)|\ge |I|$.
\end{remark}

If we apply Theorem~\ref{thm:gen-hall} to an order $\ell$ pattern, the resulting pattern $\cS'$ will not be order $\ell$ in general (in fact, it will be order $k$). To extend an order $\ell$ generic zero pattern to a maximal order $\ell$ generic zero pattern, we need a further generalization of the generalized Hall's theorem (Theorem~\ref{thm:ell-hall}). First, we state an equivalence.

\begin{proposition}\label{prop:gzp-ell} Assume $n \ge k$. Let $A_1, \hdots, A_{\ell} \subseteq [n]$ of size at most $k$. The following are equivalent.
  \begin{enumerate}
  \item[(a)] There exist $\delta_1, \hdots, \delta_\ell \ge 0$. such that for all nonempty $I \subseteq [\ell]$
    \begin{align} |A_I| \le k - \sum_{i \in I} \delta_i.\label{eq:ell-hall}
    \end{align}
  \item[(b)] The pattern $(S_1, \hdots, S_k)$, with $\delta_i$ copies of $A_i$ for $i \in [\ell]$ and additional $k-\sum_{i=1}^\ell \delta_i$ copies of the empty set, is a generic zero pattern order $\ell$.
  \end{enumerate}
\end{proposition}

\begin{proof}

  First we prove that (a) implies (b). Let $d = k - \sum_{i \in [\ell]} \delta_i$. By (\ref{eq:ell-hall}), we know that $d$ is nonnegative. We need to show that (\ref{eqn:GZP_hall}) holds for $(S_1, \hdots, S_k)$. Let $I \subseteq [k]$ be any non-empty subset, we want to show that (\ref{eqn:GZP_hall}) holds for $I$. If $I$ includes at least one $i$ for which $S_i = \emptyset$, then (\ref{eqn:GZP_hall}) trivially holds.  Further, if $I \subseteq [k]$ \change{has the property that there exists $i \in I$, $i' \not\in I$, and $j \in [\ell]$ such that $S_i = S_{i'}= A_j$, then} we might as well \change{add $i'$ to $I$} as the LHS of (\ref{eqn:GZP_hall}) is unchanged and the RHS can only decrease. \change{That is, for each $i \in [\ell]$, the multiset $\{S_j : j \in I\}$ has $0$ or $\delta_i$ copies of $A_i$.} \change{In that case,} (\ref{eqn:GZP_hall}) is exactly (\ref{eq:ell-hall}) for the relevant set of $A_i$'s.

  To prove that (b) implies (a), for every $I \subseteq [\ell]$ for (\ref{eq:ell-hall}), look at the subset $I' \subseteq [k]$ \change{for} which \change{$\{S_j : j \in I'\}$} includes $\delta_i$ copies of $A_i$ for each $i \in I$. The truth of (\ref{eq:ell-hall}) is then implied by applying (\ref{eqn:GZP_hall}) to $I'$.
\end{proof}

We now state our generalized Hall's theorem which is a further generalization of Theorem~\ref{thm:gen-hall}.

\begin{theorem}[Generalized Hall's Theorem (new)]
\label{thm:ell-hall} Assume $n \ge k$\change{. Let} $A_1, \hdots, A_{\ell}$ be subsets of $[n]$ of size at most $k$. Assume there exist $\delta_1, \hdots, \delta_\ell \ge 0$ such that for all nonempty $I \subseteq [\ell]$, (\ref{eq:ell-hall}) holds.  Then, there exists $A'_i \supseteq A_i$ such that (\ref{eq:ell-hall}) holds for $A_1',\dots,A_\ell'$ and for all $i \in [\ell]$, $|A'_i| = k - \delta_i$.
\end{theorem}
\begin{proof} It suffices to show that if for some $i \in [\ell]$, we have that $|A_i| < k- \delta_i$, then we can find $A'_i \supset A_i$ for which $|A'_i| = k - \delta_i$ and replacing $A'_i$ with $A_i$ still satisfies (\ref{eq:ell-hall}). WLOG, by permutation of the $A_i$'s, we may assume that $|A_1| < k - \delta_1$. Note that if $\delta_1 = 0$, we may extend $A_1$ to an arbitrary superset of size $k$\change{.} Since (\ref{eq:ell-hall}) holds for all $I$ with $1 \not\in I$, adding $1$ to $I$ can only decrease the LHS but not change the RHS. Thus, this is a valid extension.

  We now assume $\delta_1 \ge 1$. By Proposition~\ref{prop:gzp-ell}, we have that $S_1, \hdots, S_k$, with $\delta_i$ copies of $A_i$ for each $i \in [\ell]$ and the rest the empty set is a generic zero pattern. For all $i \in [\ell]$, let $J_i \subseteq [k]$ (possibly empty) be the indices for $j$ for which $S_j = A_i$. Let $T$ be an arbitrary subset of $[n]$ of size $k$ for which $A_1 \subseteq T$. Define $T_i := S_i \cap T$ for all $i \in [k]$.  Note that $(T_1, \hdots, T_k)$ is a generic zero pattern, \change{since} for each \change{condition required by (\ref{eqn:GZP_hall}),} the LHS could only decrease going from $S_i$ to $T_i$.

  Note that the ``$T$-complements'' $\bar{T}_i := T \setminus T_i$ satisfy the matching conditions of classical Hall's theorem (e.g.,~\cite{van2001course}). That is, for any nonempty $I \subseteq [k]$,
  \[ \left|\bigcup_{i \in I} \bar{T}_i\right| = |T| - \left|\bigcap_{i \in I} T_i\right| \ge k - (k - |I|) = |I|.
  \] Thus, by Hall's theorem, the elements of $T$ can be listed as $(t_1, \hdots, t_k)$ such that $t_i \in \bar{T}_i$ for all $i$. Define $T'_i = T \setminus \{t_i\}$. Note that $ T_i \subseteq T'_i \subset T$ for all $i \in [k]$.  Since each $T'_i$ excludes a distinct element, the family $(T'_1, \hdots, T'_k)$ satisfies (\ref{eqn:GZP_hall}).  For all $i \in [\ell]$, let
  \[ U_i := \begin{cases} \bigcap_{j \in J_i} T'_j & J_i\text{ nonempty}\\ T & \text{otherwise}.
      \end{cases}
  \]

  Note that $A_i \cap T \subseteq U_i$ for all $i$. For $i = 1$, since $A_1 \subseteq T$, so $A_1 \subseteq U_1$. By (\ref{eqn:GZP_hall}), we have that $|U_1| \le k - \delta_1$. Further, since $|T \setminus T'_i| = 1$ for all $i \in J_1$, we can deduce by the pigeonhole principle that $|U_1| = k - \delta_1$.
  
  We claim that upon replacing $A_1$ with $U_1$, (\ref{eq:ell-hall}) is still satisfied for all nonempty $I$. Since we keep the other sets unchanged, (\ref{eq:ell-hall}) holds for all $I$ with $1 \not\in I$. Now, assume that $1 \in I$. By applying (\ref{eqn:GZP_hall}) with $\bigcup_{i \in I} J_i$, we have that {\allowdisplaybreaks \begin{align*} k - \sum_{i \in I} \delta_i &= k - \sum_{i \in I} |J_i| \\ &\ge \left|\bigcap_{i \in I}\bigcap_{j \in J_i} T'_j\right|\\ &= \left|U_1 \cap \bigcap_{i \in I \setminus \{1\}} U_i\right|\\ & \ge \left|U_1 \cap \bigcap_{i \in I \setminus \{1\}} (A_i \cap T)\right|\\ &= \left|(U_1 \cap T) \cap \bigcap_{i \in I \setminus \{1\}} A_i\right|\\ &= \left|U_1 \cap \bigcap_{i \in I \setminus \{1\}} A_i\right|,
  \end{align*}} as desired. Thus, extending $A_1$ to $U_1$ was valid.

  By recursively applying the argument, we may deduce that each $A_i$ can be extended to a set of size $k - \delta_i$.
\end{proof}

 If a generic pattern $(S_1, \hdots, S_k)$ contains $\delta_i$ copies of a set $A_i$, then (\ref{eqn:GZP_hall}) implies that $|A_i| \le k-\delta_i$.  As an immediate corollary of Theorem~\ref{thm:ell-hall} and Proposition~\ref{prop:gzp-ell}, we can now show that any generic zero pattern of order $\ell$ can be extended to a ``maximal'' generic zero pattern of order $\ell$.
\begin{corollary}\label{cor:gzp-ell} Assume $n \ge k$. Let $(S_1, \hdots, S_k)$ be a generic zero pattern order $\ell$ containing $\delta_i$ copies of nonempty $A_i$ for all $i \in [\ell]$. Then, there exists $A'_i \supseteq A_i$ with $|A'_i| = k - \delta_i$ for all $\change{i \in [\ell]}$, such that the zero pattern $(S'_1, \hdots, S'_k)$ with $\delta_i$ copies of $A'_i$ for all $i$ (and the same number of copies of the empty set) is a generic zero pattern of order $\ell$.
\end{corollary}

\begin{remark} Following the methods of \cite{dau2014gmmds}, Corollary~\ref{cor:gzp-ell}, as well as the other results in this subsection can be made into efficient algorithms.
\end{remark}

\subsection{A Characterization of Sets in Order-\texorpdfstring{$\ell$}{ell} Generic Zero Patterns}\label{subsec:partition} Given subsets $A_1,A_2,\dots,A_\ell\subset [n]$, we are interested in knowing if there exists an order-$\ell$ generic zero pattern $\cS=(S_1,S_2,\dots,S_k)$ for $k\times n$ matrices which is formed by copies of $A_1,A_2,\dots,A_\ell$ and $d$ copies of the empty set.  We now give a combinatorial characterization of sets $A_1,A_2,\dots,A_\ell$ for which this can be done, based on partitions of the $\ell$ sets. This characterization is essential for relating generic zero patterns with the generic intersection problem (see Section~\ref{sec:gzp-mds} for its application).

\begin{lemma}\label{lem:hall-part} Assume $n \ge k$ and $d\ge 0$. Let $A_1, A_2, \hdots, A_\ell \subseteq [n]$ of size at most $k$. The following are equivalent.
  \begin{enumerate}
    \item[(a)] There exists an order at most $\ell$ generic zero pattern for $k\times n$ matrices which contains only copies of $A_1,A_2,\dots,A_\ell$ and an additional $d$ copies of the empty set.
    \item[(b)] There exist $\delta_1, \hdots, \delta_{\ell} \ge 0$ such that $\sum_{i=1}^\ell \delta_i = k - d$ and for all nonempty $I \subseteq [\ell]$,
      \begin{align} |A_I| \le k - \sum_{i \in I} \delta_i.\label{eq:pa}
      \end{align}
    \item[(c)] For all partitions $P_1 \sqcup P_2 \sqcup \cdots \sqcup P_{s} = [\ell]$, we have that
      \begin{align} \sum_{i=1}^s |A_{P_i}| \le (s-1)k + d.\label{eq:pb}
      \end{align}
  \end{enumerate}
\end{lemma}

\begin{proof} We first prove that (a) iff (b). If (a) is true, let $\delta_i$ be the number of times that $A_i$ appears in the generic zero pattern. Note that the empty set must then appear $d = k - \sum_{i=1}^\ell \delta_i$ times. Thus, by Proposition~\ref{prop:gzp-ell}, we have that (b) holds.

If (b) is true, then by Proposition~\ref{prop:gzp-ell}, there is a generic zero pattern which uses $A_1, \hdots, A_{\ell}$ as its nonempty sets with precisely additional $k - \sum_{i=1}^\ell \delta_i = d$ empty sets.

Now we prove that (b) iff (c). The proof that (b) implies (c) is near-immediate, for any partition $P_1 \sqcup P_2 \sqcup \cdots \sqcup P_s = [\ell]$, we have that by (\ref{eq:pa}),
\begin{align*} \sum_{i=1}^s |A_{P_i}| &\le \sum_{i=1}^s \left(k - \sum_{j \in P_i} \delta_j\right)\\ &= sk - \sum_{j=1}^{\ell} \delta_j\\ &= (s-1)k + d.
\end{align*}

The proof that (c) implies (b) is rather nontrivial and requires a careful induction on $\ell$.\footnote{Although our induction is purely combinatorial, it has some similarities to \cite{yildiz2019gmmds}'s proof of the GM-MDS theorem.} The base case $\ell = 1$ follows by taking $\delta_1 = k - d$ and the fact that $|A_1| \le d$.

Now assume that $\ell \ge 2$. First, by applying the discrete partition $\{1\} \sqcup \{2\} \sqcup \cdots \sqcup \{\ell\}$ to (\ref{eq:pb}), we know that
\begin{align} \sum_{i = 1}^{\ell} |A_i| \le (\ell-1)k + d\label{eq:p1},
\end{align} which is \change{at most $nk$}. Thus, we can pad the $A_i$'s with additional elements \change{while maintaining that each $A_i$ has size at most $k$} until (\ref{eq:pb}) is an equality for some partition $P_1 \sqcup P_2 \sqcup \cdots \sqcup P_s = [\ell]$. We call such a partition \emph{tight}. Note that padding the $A_i$'s with additional elements can only make (\ref{eq:pa}) more difficult to satisfy, and thus we only need to prove \change{(c)} implies \change{(b)} for the padded family of sets.

If the tight partition satisfies $s = 1$, that is, $|A_{[\ell]}| = d$, \change{we} claim that we can continue to pad\footnote{A similar padding argument appears in Appendix B of \cite{bgm2021mds}.} the $A_i$'s until a partition with $s \ge 2$ is tight. If (\ref{eq:p1}) is tight, this is already true. Otherwise, we have that
\[ \sum_{i = 1}^{\ell} |A_i \setminus A_{[\ell]}| < (\ell-1)(k-d)\change{.}
\]
\change{Let $I \subseteq [\ell]$ be the set of indices for which $|A_i| = k$ and so $|A_i \setminus A_{[\ell]}| = k-d$. These are the sets which we can no longer pad. From the previous equation, we can deduce that
\[
\sum_{i \in [\ell] \setminus I} |A_i \setminus A_{[\ell]}| < (\ell-|I|-1)(k-d).
\]
}
Thus, by the pigeonhole principle, there is $\change{j} \in [n]\setminus A_{[\ell]}$ \change{for which} 
\[
\change{\{i \in [t] \setminus I : j \in A_i\} \le \lfloor ((\ell-|I|-1)(k-d) - 1) / (n-d)\rfloor \le \ell - |I| - 2.}
\]
Thus, we can continue padding the $A_i$'s, while keeping $|A_{[\ell]}| = d$ \change{and maintaining that each $A_i$ has size at most $k$}, until a partition $P_1 \sqcup P_2 \sqcup \cdots \sqcup P_s = [\ell]$ with $s \ge 2$ is tight.

Fix $i \in [s]$. For all $j \in P_i$, define $B_j = A_j \setminus A_{P_i}$. Let $k_i = k - |A_{P_i}|$. Since each $A_j$ has size at most $k$, each $B_j$ has size at most $k_i\le k$. We claim that for all partitions $Q_1 \sqcup Q_2 \sqcup \cdots \sqcup Q_t = P_i$, we have that
\begin{align} \sum_{j=1}^t |B_{Q_j}| \le (t-1) k_i\label{eq:qb}
\end{align} To see why, note that $P_1 \sqcup \cdots \sqcup P_{i-1} \sqcup Q_1 \sqcup \cdots \sqcup Q_t \sqcup P_{i+1} \sqcup \cdots \sqcup P_{s} = [\ell]$ is a partition. Thus we may apply (\ref{eq:qb}) on this partition and the fact that $P_1 \sqcup P_2 \sqcup \cdots \sqcup P_s$ is a tight partition to obtain that.
\begin{align*} \sum_{j=1}^t |B_{Q_j}| &= \sum_{j=1}^t |A_{Q_j}| - t|A_{P_i}|\\ &= \left(\sum_{j=1}^t |A_{Q_j}| + \sum_{j \in [s] \setminus i} |A_{P_j}|\right) - \sum_{j=1}^s |A_{P_j}| - (t-1)|A_{P_i}|\\ &\le (s+t-2)k + d - ((s-1)k+d) - (t-1)(k - k_i)\\ &= (t-1)k_i.
\end{align*}

Since $s \ge 2$, we have that $|P_i| < \ell$. Thus, we may apply the induction hypothesis\footnote{Note that recursive padding in this induction as the ``output'' of the recursion is the $\delta_i$'s satisfying (\ref{eq:pa}) and not the subsequent modifications to the sets.} to $(B_j : j \in P_i)$ to get that there exist $\delta_j \ge 0$ for all $j \in P_i$ such that $\sum_{j \in P_i} \delta_j = k_i$ and for all nonempty $J \subseteq P_i$,
\begin{align} |B_J| \le k_i - \sum_{j \in J} \delta_j \label{eq:qa}
\end{align}

We can perform this procedure for all $i \in [s]$, and thus obtain a $\delta_j$ for all $j \in [\ell]$ via that $i \in [s]$ for which $j \in P_i$. We claim that these exact same $\delta_j$'s satisfy condition (a) for $A_1, \hdots, A_\ell$. First, observe that
\[ \sum_{j=1}^\ell \delta_j = \sum_{i=1}^s \sum_{j\in P_i} \delta_j = \sum_{i=1}^s k_i = \sum_{i=1}^s (k - |A_{P_i}|) = k - d,
\] where the last equation uses that the partition is tight.

Now we verify (\ref{eq:pa}). Pick nonempty $I \subseteq [\ell]$. Let $\sigma \subseteq [s]$ be the set of indices $i$ for which $I \cap P_i$ is nonempty. Then, {\allowdisplaybreaks
\begin{align*} |A_I| &= \left|\bigcap_{i \in \sigma} \bigcap_{j \in I \cap P_i} A_j\right|\\ &= \left|\bigcap_{i \in \sigma} \left(A_{P_i} \cup \bigcap_{j \in I \cap P_i} B_j\right)\right|\\ &\le \left|\bigcap_{i \in \sigma} A_{P_i}\right| + \sum_{i \in \sigma} |B_{I \cap P_i}|\\ &= \left|\bigcap_{i \in \sigma} A_{P_i}\right| + \sum_{i \in \sigma} \left(k_i - \sum_{j \in I \cap P_i} \delta_j\right)\\ &= \left|\bigcap_{i \in \sigma} A_{P_i}\right| - \sum_{i \in \sigma} |A_{P_i}| + k|\sigma| - \sum_{j \in I} \delta_i\\ &= \left|\bigcap_{i \in \sigma} A_{P_i}\right| + \sum_{i \not\in \sigma} |A_{P_i}| - \sum_{i \in [s]} |A_{P_i}| + k|\sigma| - \sum_{j \in I} \delta_i\\ &\le ((s-|\sigma|+1)-1)k + d - ((s-1)k + d)+ k|\sigma| - \sum_{j \in I} \delta_i\\ &= k - \sum_{j \in I} \delta_i,
\end{align*} } where the last inequality follows from applying (\ref{eq:pb}) to the partition $\{\bigcup_{i \in \sigma} P_i\} \cup \{P_j : j \not\in \sigma\}$.
\end{proof}

\section{\texorpdfstring{Equivalence of $\GZP(\ell)$ and $\MDS(\ell)$}{Equivalence of GZP(ell) and MDS(ell)}}\label{sec:gzp-mds}

In this section, we prove \change{two of the main results. First, we prove Theorem~\ref{thm:main-dim}} on computing the dimension of a generic intersection\change{.}

\change{\dimgen*}

\change{Second, we show} that $\MDS(\ell)$ is equivalent to $\GZP(\ell)$ for all $\ell \ge 1$ (and thus that (a) and (c) are equivalent in Theorem~\ref{thm:main-mds}).

\change{\mdsequiv*}

\subsection{\texorpdfstring{$\GZP(\ell)$ Implies $\MDS(\ell)$}{GZP(ell) Implies MDS(ell)}}
\label{sec:gzp_implies_mds}

The key result toward proving that $\GZP(\ell)$ implies $\MDS(\ell)$ is the following:

\begin{lemma}\label{lem:gzp-dim} Let $C$ be an $\change{[n,k]}$-code whose generator matrix $\change{G \in \F^{k \times n}}$ is $\GZP(\ell)$. Let $A_1, \hdots, A_\ell \subseteq [n]$ be sets of size at most $k$. Then,
  \begin{align} \dim (G_{A_1} \cap \cdots \cap G_{A_{\ell}}) = \max_{P_1\sqcup P_2 \sqcup \dots \sqcup P_s=[\ell]} \left(\sum_{i\in [s]} \left|\bigcap_{j \in P_i} A_j\right| - (s-1) k\right) \label{eq:gzp-dim}
  \end{align} where the maximum is over all partitions of $[\ell].$
\end{lemma}

By Proposition~\ref{prop:our-gm-mds}, a generic matrix is $\GZP(\ell)$ for all $\ell$ and thus Theorem~\ref{thm:main-dim} is an immediate corollary of Lemma~\ref{lem:gzp-dim}. As a result, replacing the RHS of (\ref{eq:gzp-dim}) with the LHS of (\ref{eq:gen-dim}), we have that any $\GZP(\ell)$ code is also $\MDS(\ell)$.

\begin{proof}[Proof of Lemma~\ref{lem:gzp-dim}] We prove (\ref{eq:gzp-dim}) as two inequalities. First we show that $\ge$ direction of (\ref{eq:gzp-dim}).  Observe that for any $S \subset [\ell]$, $G_{A_S} \subset \bigcap_{j \in S} G_{A_j}$. Thus, for every partition $P_1\sqcup P_2 \sqcup \dots \sqcup P_s=[\ell]$, we have that
  \[ \dim \bigcap_{i \in [s]} G_{A_{P_i}} \le \dim \bigcap_{j \in [\ell]} G_{A_j}.
  \] Thus, since $G$ is MDS,\footnote{See a similar argument in \cite{bgm2021mds}.}
  \[ \dim \bigcap_{j \in [\ell]} G_{A_j} \ge \dim \bigcap_{i\in [s]} G_{A_{P_i}} \ge k - \sum_{i \in [s]} \dim G^{\perp}_{A_{P_i}} = k - \sum_{i \in [s]} (k - |A_{P_i}|) = \sum_{i \in [s]} |A_{P_i}| - (s-1)k,
  \] as desired.

  Now we show that $\le$ direction. Let $d$ be the RHS of (\ref{eq:gzp-dim}); that is the minimum choice of $d$ such that condition (b) in Lemma~\ref{lem:hall-part} holds. Thus, by the lemma, there exist $\delta_1, \hdots, \delta_\ell \ge 0$ such that $\sum_{i=1}^\ell \delta_i = k - d$ and for all nonempty $I \subseteq [\ell]$,
  \begin{align*} |A_I| \le k - \sum_{i \in I} \delta_i.
  \end{align*} By Proposition~\ref{prop:gzp-ell}, the pattern $\cS := (S_1, \hdots ,S_k)$ with $\delta_i$ copies of $A_i$ for each $i \in [\ell]$ and $d$ copies of the empty set (assume $S_{k-d+1} = \cdots = S_k = \emptyset$) is a generic zero pattern of order $\ell$. Thus, since $C$ is a $\GZP(\ell)$ code, there exists an invertible matrix $\change{M \in \F^{k \times k}}$ such that $MG$ has the zero pattern $\cS$.

  Note that $\dim((MG)_{A_1} \cap \cdots \cap (MG)_{A_{\ell}}) = \dim (G_{A_1} \cap \cdots \cap G_{A_{\ell}})$. Thus, in order to show that $\dim (G_{A_1} \cap \cdots \cap G_{A_{\ell}}) \le d$, it suffices to show that for any $z \in (MG)_{A_1} \cap \cdots \cap (MG)_{A_{\ell}}$, only the last $d$ coordinates can be nonzero.

  For each $i \in [\ell]$, let $I_i$ be the indices $j$ for which $S_j = A_i$. Since $z \in (MG)_{A_i}$, and $MG$ has the zero pattern $\cS$, we have that $z|_{I_i} = 0$. Thus,
  \[z|_{\bigcup_{i \in [\ell]} I_i} = 0.\] Thus, $z$ is only nonzero on the coordinates corresponding to empty $S_i$'s in the partition. That is, only the last $d$ coordinates of $z$ can be nonzero, proving the dimension upper bound.
\end{proof}

\subsection{\texorpdfstring{$\MDS(\ell)$ Implies $\GZP(\ell)$}{MDS(ell) Implies GZP(ell)}}

The ``(a) implies (c)'' part of Theorem~\ref{thm:main-mds} follows from the following lemma.

\begin{lemma} Let $C$ be a $\change{[n,k]}$-code which is $\MDS(\ell)$.  Let $\change{G \in \F^{k \times n}}$ be its generator matrix. Let $\cS= (S_1, \hdots, S_k)$ be a generic zero pattern of order $\ell$. Then, $G$ attains $\cS$.
\end{lemma}

\begin{proof} By Corollary~\ref{cor:gzp-ell}, there exists a maximal order-$\ell$ generic zero pattern $\cS'$ which is a superset of $\cS$. That is, there are $A_1, \hdots, A_\ell$ and $\delta_1, \hdots, \delta_\ell$, such that $\cS'$ is the union for all $i \in [\ell]$ of $\delta_i$ copies of $A_i$ plus $d := k - \sum_{i=1}^\ell \delta_i$ copies of the empty set. Further, by ``maximality'', each $A_i$ is of size $k - \delta_i$. It suffices to show that $C$ attains $\cS'$.

  Because $\cS'$ is a generic zero pattern, we have for all nonempty $I \subseteq [\ell]$,
  \[ |A_I| \le k - \sum_{i \in I} \delta_i.
  \] Thus, by Lemma~\ref{lem:hall-part}, we have that for all partitions $\cP$ of $[\ell]$, we have that
  \[ \sum_{P \in \cP} |A_P| \le (|\cP|-1)k + d.
  \] For the partition, $\{1\} \cup \cdots \cup \{\ell\}$, we have equality because each $|A_i| = k - \delta_i$.

  Thus, since $C$ is $\MDS(\ell)$, by Theorem~\ref{thm:main-dim}, we have that
  \[ \dim(G_{A_1} \cap \cdots \cap G_{A_\ell}) = d.
  \] Let $X := G_{A_1} \cap \cdots \cap G_{A_\ell}$ be this $d$-dimensional subspace of $\F^k$. Thus, by taking the dual,
  \[ X^{\perp} = G_{A_1}^{\perp} + \cdots + G_{A_\ell}^{\perp}.
  \] Note that $\dim(G_{A_i}^{\perp}) = k - |A_i| = \delta_i$. And $\dim(X^{\perp}) = k - d = \sum_{i \in [\ell]} \delta_i$.  Therefore, if for each $i \in [\ell]$, we pick a basis $(v^{i}_1, \hdots, v^{i}_{\delta_i})$ for $G_{A_i}^{\perp}$, then $(v^i_j : i \in [\ell], j \in [\delta_i])$ is a basis for $X^{\perp}$.

  Let $\change{M \in \F^{k \times k}}$ be an arbitrary invertible matrix whose first $k-d$ rows are $(v^i_j : i \in [\ell], j \in [\delta_i])$. We claim that $MG$ has zero pattern $\cS$. For each $i \in [\ell]$ and $j \in [\delta_i]$, let $i'$ be the row of $M$ corresponding to $v^i_j$. For any $j' \in [n]$, we have that $(MG)_{i',j'}$ is the inner product\footnote{Here, ``inner product'' is the informal name for the bilinear form $\langle v, w\rangle = \sum_{i=1}^k v_i w_i.$} of $v^i_j$ and $G_{j'}$. If $j' \in A_i$, then $G_{j'} \in G_{A_i}$. Since $v^i_j \in G_{A_i}^{\perp}$, the inner product is $0$.

  Thus, $MG$ indeed has the desired zero pattern.
\end{proof}

\subsection{Characterizing the Null Intersection Property}

We recall from \cite{bgm2021mds} the following definition.

\begin{definition} Let $A_1, \hdots, A_\ell \subseteq [n]$ be sets of size at most $k$. We say that $A_1, \hdots, A_\ell$ have the \emph{null intersection} property if for a generic matrix, $\change{W \in \F^{k \times n}}$ we have that
  \[ W_{A_1} \cap \cdots \cap W_{A_k} = 0.
  \]
\end{definition}

Note that Theorem~\ref{thm:main-dim} characterizes which sets have the null intersection property: they are those for which the RHS of (\ref{eq:gen-dim}) is equal to $0$. Notably, \cite{bgm2021mds} also shows that the $\MDS(\ell)$ condition is equivalent to attaining null intersection for any family of $\ell$ sets with the null intersection property:
\begin{lemma}[Lemma 3.1~\cite{bgm2021mds}]\label{lem:null} Let $C$ be a $\change{[n,k]}$-code with generator matrix $G$. Let $\ell \ge 2$. The following are equivalent.
\begin{enumerate}
\item[(a)] $C$ is $\MDS(\ell)$.
\item[(b)] $C$ is $\MDS$ and for all $A_1, \hdots, A_\ell \subseteq [n]$ of size at most $k$ and $|A_1| + \cdots + |A_\ell| = (\ell-1)k$, we have that
  \[ G_{A_1} \cap \cdots \cap G_{A_{\ell}} = 0
  \] if and only if it holds generically.
\end{enumerate}
\end{lemma}

\begin{remark} Composing this lemma and the proof that $\MDS(\ell)$ and $\GZP(\ell)$ are equivalent, one can in fact show that, for $\ell \ge 2$, $\GZP(\ell)$ is equivalent to the seemingly weaker property that the matrix $G$ attains all generic zero patterns with at most $\ell$ distinct sets \emph{including} the empty set. We omit further details.
\end{remark}

\section{Applications to List Decoding: Proof of Theorem~\ref{thm:main-ld}}\label{sec:ld-mds}

In this section, we show \change{that (a) and (b) are equivalent in Theorem~\ref{thm:main-mds}}.

\change{\mdsequiv*}

\change{In other words, we show} that the recently established notions of higher-order MDS codes in the works of \cite{bgm2021mds} and \cite{roth2021higher} are in fact equivalent.  An equivalent way to define $\LDMDS(L)$ codes is using the parity check matrix $\change{H \in \F^{(n-k)\times n}}$ matrix of $C.$ $C$ is $\LDMDS(L)$ if there \change{do not} exist $L+1$ \emph{distinct} vectors $u_0,u_1,\dots,u_L \in \F^n$ such that
\begin{align}
	\label{eqn:LDMDS_paritycheck} \sum_{i=0}^L \wt(u_i) \le L(n-k) \text{ and }Hu_0=Hu_1=\dots=Hu_L.
\end{align}

\subsection{Equivalence of \texorpdfstring{$\MDS$}{MDS} and \texorpdfstring{$\LDMDS$}{LDMDS} (up to Duality)}

We now prove that (a) iff (b) in Theorem~\ref{thm:main-mds}.  We will break the proof into Propositions~\ref{prop:LDMDS_implies_MDS} and \ref{prop:MDS_implies_LDMDS}.

\begin{proposition}
	\label{prop:LDMDS_implies_MDS} $C$ is $\LDMDS(\le L) \Rightarrow$ $C^\perp$ is $\MDS(L+1).$
\end{proposition}
\begin{proof} We will prove the contrapositive: $C^\perp$ is not $\MDS(L+1) \Rightarrow$ $C$ is not $\LDMDS(\le L)$.

	Let $C$ be a $\change{[n,k]}$-code and $\change{H \in \F^{(n-k)\times n}}$ be its parity check matrix. Note that $H$ is the generator matrix of $C^\perp.$ By Lemma~\ref{lem:null}, the fact that $C^\perp$ is not $\MDS(L+1)$ implies that there exists subsets $J_0,J_1,\dots,J_L \subset [n]$ such that
	\begin{enumerate}
		\item $H_{J_0} \cap H_{J_1} \cap \dots \cap H_{J_L} \ne 0$
		\item $W_{J_0} \cap W_{J_1} \cap \dots \cap W_{J_L} = 0$ for a generic \change{$(n-k) \times n$ matrix $W.$}
	\end{enumerate} By Theorem~\ref{thm:main-dim}, (2) implies that for all partitions $P_1 \sqcup P_2 \sqcup \dots \sqcup P_s = \{0,1,\dots,L\}$, we have $\sum_{i=1}^s |J_{P_i}| \le (s-1)(n-k)$ where $J_{P_i}=\bigcap_{j \in P_i} J_j.$ (1) implies that there exist non-zero $u_0,u_1,\dots,u_L\in \F^n$ such that $$\supp(u_i)\subset J_i \text{ and } Hu_0=Hu_1=\dots=Hu_L.$$ Suppose there are $s$ distinct vectors among $u_0,u_1,\dots,u_L.$ Let $P_1 \sqcup P_2 \sqcup \dots \sqcup P_s = \{0,1,\dots,L\}$ be the partition of $L+1$ into $s$ parts such that all the $\{u_j: j\in P_i\}$ are equal for every $i\in [s]$ (and they are distinct if they fall in different parts). Let $u_{P_i}$ be the common vector equal to $\{u_j: j \in P_i\}.$ Note that $\supp(u_{P_i}) \subset \bigcap_{j\in P_i} J_j = J_{P_i}$. Therefore we have $s$ distinct non-zero vectors $u_{P_1},u_{P_2},\dots, u_{P_s}$ such that
	$$\sum_{i=1}^s \wt(u_{P_i})\le\sum_{i=1}^s |J_{P_i}| \le (s-1)(n-k) \text{ and } Hu_{P_1}=Hu_{P_2}=\dots=Hu_{P_s}.$$ If $s=1$, we get $\wt(u_{P_1})\le 0$ which is not possible since $u_{P_1}$ is non-zero. Therefore $s\ge 2$ and this violates $\LDMDS(s-1)$ and therefore $\LDMDS(\le L).$
\end{proof}

\begin{proposition}
	\label{prop:MDS_implies_LDMDS} $C^\perp$ is $\MDS(L+1) \Rightarrow$ $C$ is $\LDMDS(\le L)$
\end{proposition}
\begin{proof} Since $\MDS(\ell)\implies \MDS(\le \ell)$, it is enough to show that $C^\perp$ is $\MDS(\ell) \Rightarrow$ $C$ is $\LDMDS(\ell-1)$ for $\ell \ge 2.$ We will prove the contrapositive: $C$ is MDS and $C$ is not $\LDMDS(\ell-1) \Rightarrow$ $C^\perp$ is not $\MDS(\ell).$

\change{We prove this by induction on $\ell$. The base case of $\ell=2$ follows from the fact that $\MDS(2)$ and $\LDMDS(1)$ are both equivalent to MDS (Section 1.2 of \cite{bgm2021mds} and Section I.B of \cite{roth2021higher}).}

\change{We may now assume that $\ell \ge 3$.} \change{Assume} $C$ \change{is} an MDS $\change{[n,k]}$-code \change{(or else $C^{\perp}$ is not even $\MDS(2)$)} and \change{let} $\change{H \in \F^{(n-k)\times n}}$ be its parity check matrix. Note that $H$ is the generator matrix of $C^\perp.$ Since $C$ is not $\LDMDS(\ell-1)$, there exist \emph{distinct} $u_1,u_2,\dots,u_\ell \in \F^n$ such that
	$$\sum_{i=1}^\ell \wt(u_i) \le (\ell-1)(n-k) \text{ and } Hu_1=Hu_2=\dots=Hu_\ell.$$ Let $J_i=\supp(u_i).$ \change{If for some $i \in [\ell]$, we have that $|J_i| \ge n-k$, then we can observe that $\{u_1, \hdots, u_\ell\} \setminus \{u_i\}$ certifies that $C$ is not $\LDMDS(\ell-2)$. Therefore, by induction, $C^{\perp}$ is neither $\MDS(\ell-1)$ nor $\MDS(\ell)$.}

\change{Therefore, we can assume that for all $i \in [\ell]$ that $|J_i|< n-k.$} We can also infer that \change{$Hu_i$ is non-zero for all $i \in [\ell]$. Otherwise, $Hu_1 = Hu_2 = \cdots Hu_\ell = 0$. Since $H$ is $\MDS$, every $n-k$ columns of $H$ are linearly independent. Therefore, for all $i \ge 1$, either $u_i = 0$ or $\wt(u_i)\ge n-k+1$.} \change{Since at most one of the $u_i$'s can be equal to zero, these bounds violate} the $\sum_{i=1}^\ell \wt(u_i)\le (\ell-1)(n-k)$ condition. \change{Therefore, we assume for all $i \in [\ell]$ that $Hu_i$ is nonzero.} The following claim completes the proof, since it proves that $C^\perp$ is not $\MDS(\ell).$
	\begin{claim} $\dim(H_{J_1} \cap H_{J_2} \cap \dots \cap H_{J_\ell}) > \dim(W_{J_1} \cap W_{J_2} \cap \dots \cap W_{J_\ell})$ for a generic \change{$(n-k)\times n$ matrix $W.$}
	\end{claim}
	\begin{proof} If $\dim(W_{J_1} \cap \dots \cap W_{J_\ell})=0$, then we are done because $Hu_1\ne 0$ and $Hu_1 \in H_{J_1} \cap \dots \cap H_{J_\ell}.$ Therefore assume that $\dim(W_{J_1} \cap W_{J_2} \cap \dots \cap W_{J_\ell})>0.$ By Theorem~\ref{thm:main-dim}, $$\dim(W_{J_1} \cap \dots \cap W_{J_\ell})= \sum_{i=1}^s |J_{P_i}| - (s-1)(n-k)$$ for some partition $P_1\sqcup \dots \sqcup P_s = [\ell]$. Note that $s<\ell$, since $\sum_{i=1}^\ell |J_i| - (\ell-1)(n-k)\le 0.$ Let $$V = \{(x_1,x_2,\dots,x_\ell) : Hx_1=Hx_2=\dots=Hx_\ell \text{ and } \supp(x_i)\subset J_i\}$$ which is a subspace of $(\F^{n})^\ell$.  Clearly $(u_1,u_2,\dots,u_\ell)\in V.$ Define the linear map $f:H_{J_1} \cap H_{J_2} \cap \dots \cap H_{J_\ell} \to V$ as: $f(y)=(x_1,x_2,\dots,x_\ell)$ where $x_i$ is uniquely defined by $Hx_i = y$ and $\supp(x_i)\subset J_i.$ It is clear that $f$ is a bijection which implies that $$\dim(V)=\dim(H_{J_1} \cap H_{J_2} \cap \dots \cap H_{J_\ell}).$$ Define $$V_\cP = \{(x_1,x_2,\dots,x_\ell) : Hx_1=Hx_2=\dots=Hx_\ell,\ \supp(x_i)\subset J_i \text{ and } x_j=x_{j'} \text{ if }j,j'\in P_i\}.$$ Clearly $V_\cP$ is a subspace of $V$. But $(u_1,u_2,\dots,u_\ell)\notin V_\cP$ since all the $u_i$ are distinct, but the partition $\cP=P_1\sqcup \dots \sqcup P_s$ has only $s<\ell$ parts. Therefore $\dim(V)>\dim(V_\cP).$ We will now show that $\dim(V_\cP)\ge \sum_{i=1}^s |J_{P_i}| - (s-1)(n-k)$ which finishes the proof of the claim.

		Define the linear map $f_\cP : H_{J_{P_1}} \cap H_{J_{P_2}} \cap \dots \cap H_{J_{P_s}} \to V_\cP$ as: $f_{\cP}(y)=(x_1,x_2,\dots,x_\ell)$ where $x_i$ is uniquely defined by $Hx_i=y$ where $\supp(x_i)\subset J_{P_j}$ where $i\in P_j.$ Again, $f_\cP$ is also a bijection. Therefore $\dim(V_\cP)=\dim(H_{J_{P_1}} \cap \dots \cap H_{J_{P_s}})$. By Theorem~\ref{thm:main-dim}, for a generic \change{$(n-k)\times n$ matrix $W,$} $\dim(H_{J_{P_1}} \cap \dots \cap H_{J_{P_s}})\ge \dim(W_{J_{P_1}} \cap \dots \cap W_{J_{P_s}}) \ge \sum_{i=1}^s |J_{P_i}|-(s-1)(n-k)$.
\end{proof}
\end{proof}

\subsection{Reed-Solomon Codes}

\subsubsection{Generic Reed-Solomon \change{C}odes}

Our main result Theorem~\ref{thm:main-ld} follows from this next result on duals of generic Reed-Solomon codes.

\begin{proposition}
  \label{prop:genericRS_duality} The dual of a generic Reed-Solomon code is equivalent to a generic Reed-Solomon code (up-to scaling of columns of the generator matrix).
\end{proposition}
\begin{proof} Let $C$ be a generic $\change{[n,k]}$-Reed-Solomon code. Let $(\alpha_1, \hdots, \alpha_n)$ be the generators of $C$. Let $C'$ be the $\change{[n,n-k]}$-code also generated by $(\alpha_1, \hdots, \alpha_n)$. Let $G'_{n-k,n}$ be the Vandermonde matrix with generators $(\alpha_1, \hdots, \alpha_n)$. Note that by definition $G'$ is also generic.

  For all $i \in [n]$, let $\Delta_i = \prod_{j \in [n]\setminus\{i\}} (\alpha_i - \alpha_j)$. By standard results\footnote{See, for example, \cite{macwilliams1977theory} or Theorem 5.1.6 in~\cite{hall7notes}.} on Reed-Solomon codes, $C^{\perp}$ is a \emph{generalized} Reed-Solomon code with generator matrix $H$ such that $H_{j,i} = \alpha_i^{j-1} / \Delta_i$. Since $C$ is generic, we have that $\Delta_i \neq 0$ for all $i \in [n]$. Therefore, each column of $H$ is a nonzero scalar multiple of a column of $G'$. Thus $C^\perp$ is equivalent to a generic Reed-Solomon code.
\end{proof}

\begin{proof}[Proof of Theorem~\ref{thm:main-ld}] From Theorem~\ref{thm:gmmds_gzp} (GM-MDS), Theorem~\ref{thm:main-mds}, and Proposition~\ref{prop:genericRS_duality}, we know that generic Reed-Solomon codes are $\LDMDS(L)$ for all $L$. In other words they are $(\rho,L)$-average-radius list-decodable for $\rho=\frac{L}{L+1}(1-R)$ for all $L.$ This also implies that they are $(\rho,L)$-list-decodable for the same $\rho.$
\end{proof}

\subsubsection{Random Reed-Solomon \change{C}odes}
\label{sec:random_rs_construction}

Here we prove Theorem~\ref{thm:main-ld-random} and show how one can make Theorem~\ref{thm:main-ld} more quantitative in order to reason about the list-decoding capabilities of a \emph{random} Reed-Solomon code. For code parameters $\change{[n,k]}$ and a finite field $\F$, we define a random $\change{[n,k]}$-Reed-Solomon code to be one whose generators $\alpha_1, \hdots, \alpha_n$ are chosen uniformly and independently from $\F$.

\begin{proposition}\label{prop:rRS_ldmds} Let $n,k,L$ be positive integers. There exists a function $c(n,k,L) = 2Ln^2\binom{n}{\le n-k}^{L+1}$ such that a random $\change{[n,k]}$-Reed-Solomon code is $\LDMDS(\le L)$, and thus $(\frac{\ell}{\ell+1}(1-k/n),\ell)$-list decodable for all $\ell \le L$, with probability at least $1 - c / |\F|$.
\end{proposition}

\begin{proof} Let $C$ be the code generated by the random \change{choice of} $\alpha_1, \hdots, \alpha_n$. With probability at least $1 - n^2/|\F|$, we have that $\alpha_i \neq \alpha_j$ for all $i \neq j \in [n]$ and thus $C$ is $\MDS$. From the proof of Proposition~\ref{prop:genericRS_duality} and Theorem~\ref{thm:main-mds}, we then have that $(\alpha_1, \hdots, \alpha_n)$ generate \change{an} $\LDMDS(\le L)$ matrix if and only if the $(n,n-k)$-matrix $G'$ with entries $G'_{j,i} = \alpha_i^{j-1}$ is $\MDS(L+1)$. In order to check that $G'$ is $\MDS(\change{L+1})$, it suffices to show by Lemma~\ref{lem:null} that for all null intersecting families $A_1, \hdots, A_{L+1} \subseteq [n]$, each of size at most $n-k$, with total size $L(n-k)$, we have that
  \begin{align} G'_{A_1} \cap \cdots \cap G'_{A_{L+1}} = 0.\label{eq:asdf}
  \end{align}
\change{We now use the following result of \cite{bgm2021mds}.
\begin{proposition}[Proposition B.1~\cite{bgm2021mds}, c.f., \cite{tian2019formulas}]\label{prop:bgm}
For any $G \in \F^{k \times n}$ and $A_1, \hdots, A_\ell \subseteq [n]$,
\[
  \rank \begin{pmatrix} \change{G|_{A_1}} & \change{G|_{A_2}} & & &\\ \change{G|_{A_1}} & & \change{G|_{A_3}} & &\\ \vdots & & & \ddots &\\ \change{G|_{A_1}} & & & & \change{G|_{A_{\ell}}}
    \end{pmatrix}
    = \sum_{i=1}^{\ell} \dim(G_{A_i}) - \dim(G_{A_1} \cap \cdots \cap G_{A_\ell}).
\]
\end{proposition}
}
\change{Since $G'$ is MDS and each $A_i$ has size at most $n-k$, $\dim(G_{A_i}) = |A_i|$ for all $i \in [L+1]$. Further, size the total size of the $A_i$'s is $L(n-k)$, we have that (\ref{eq:asdf}) is equivalent to the following block matrix being nonsingular.}
  \[
    \begin{pmatrix} \change{G'|_{A_1}} & \change{G'|_{A_2}} & & &\\ \change{G'|_{A_1}} & & \change{G'|_{A_3}} & &\\ \vdots & & & \ddots &\\ \change{G'|_{A_1}} & & & & \change{G'|_{A_{L+1}}}
    \end{pmatrix}
  \] The square matrix has size $L(n-k)$ and each entry has degree at most $n-k-1$. Therefore, the determinant, which we know must not symbolically vanish by Corollary~\ref{cor:main-RS}, has total degree at most $Ln^2$. Therefore, by the Schwartz-Zippel lemma~\cite{Sch80,Zip79}, the probability that this determinant is zero is at most $\frac{Ln^2}{|\F|}$. Now, the number of choices\footnote{This can be optimized.}  of $A_1, \hdots, A_{L+1} \subseteq [n]$ to consider is at most $\binom{n}{\le n-k}^{L+1}$. Therefore, the probability that $G'$ is $\MDS(L+1)$, and thus $C$ is $\LDMDS(\le L)$ is at most,
  \[ 1 - \frac{n^2 + Ln^2\binom{n}{\le n-k}^{L+1}}{|\F|},
  \] as desired.
\end{proof}

\begin{remark}
  \label{rem:bgm_construction} The work \cite{bgm2021mds} shows that a random linear code achieves $\MDS(\ell)$ at field size $O_{\ell}(n^{(n-k)(\ell-1)} (n-k)^{2\ell(n-k)})$, and thus is $\MDS(\le L)$ with field size $n^{O(kL)}$. Since we prove that generic Reed-Solomon codes are $\LDMDS(\le L)$, one can adapt their methods to then show that one can take $c(n,k,L) = n^{O(\min(k,n-k)L)}$ in Theorem~\ref{thm:main-ld-random}.
\end{remark}

\subsection{Resolution of Conjecture 5.7 of \texorpdfstring{\cite{shangguan2020combinatorial}}{[ST20]}}
\label{sec:conj_alg_st20}

Shangguan and Tamo~\cite{shangguan2020combinatorial} made an algebraic conjecture in their paper (see Conjecture 5.7 from \cite{shangguan2020combinatorial}) about the non-singularity of certain symbolic matrices, which would imply that generic Reed-Solomon codes achieve list decoding capacity. In this section, we prove this conjecture. The proof follows from some of the results we use to prove Theorem~\ref{thm:main-ld}.  We first introduce the necessary notation to properly state Conjecture~5.7 of \cite{shangguan2020combinatorial}.

Let $C$ be a $\change{[n,k]}$-code with generator matrix $G$. Let $J_1, \hdots, J_t \subseteq [n]$. We define a block matrix $M_{G,(J_1,\hdots, J_t)}$ \change{with $\binom{t}{2}k$ columns and $\binom{t-1}{2}k + \sum_{i \neq j} |J_i \cap J_j|$ rows} as follows.

\change{To describe the entries of $M$, we split the columns of $M$ into $\binom{t}{2}$ blocks of size $k$, where we identify each block with a pair $\{i,j\} \in \binom{[t]}{2}$. For the rows, we have two types (a) and (b), which consist of $\binom{t-1}{2}k$ and $\sum_{i \neq j} |J_i \cap J_j|$ rows, respectively. For the (a)-type rows, we split them into $\binom{t-1}{2}$ blocks of size $k$, where identify the blocks with pairs $\{i,j\} \in \binom{[t-1]}{2}$. For the (b)-type rows, we split them into $\binom{t}{2}$ blocks, which are identified with pairs $\{i,j\} \in \binom{[t]}{2}$. The $(i,j)$th type-(b) row block has $|J_i \cap J_j$ rows. The entries of $M$ are filled in as follows.}

\begin{itemize}
\item[(a)] \change{Consider the type (a) row block indexed by $\{i,j\} \in \binom{[t-1]}{2}$ with $i < j$. Where this row meets the column blocks $\{i,j\}$ and $\{j,t\}$, we insert the identity matrix $I_k$. Where this row meets the column block $\{i,t\}$, the matrix $-I_k$ appears. All other entries in this row block are zero.}

\item[(b)] \change{Consider the type (b) row block indexed by $\{i,j\} \subseteq [t].$ Where this row block meets column block $\{i,j\}$, we insert the matrix $G^{\top}|_{J_i \cap J_j}$. All other entries in this row block are zero.}
\end{itemize}

The following lemma of \cite{shangguan2020combinatorial} relates $M_{G,(J_1,\hdots, J_t)}$ to list decoding properties of $C$.

\begin{lemma}[\cite{shangguan2020combinatorial}, lightly edited] Let \change{$t \ge 2$ and }$\tau \ge 1$. Assume there exists $\change{c_1, c_2, \hdots, c_t} \in C \cap B_{\tau}(y)$. Then, there exist $\change{J_1, \hdots, J_t}$ such that $|J_i| \ge n - \tau$ for all $i$ and $M_{G,\change{(J_1, \hdots, J_t)}}$ does not have full column rank.
\end{lemma}

To give intuition about $M$, we include an adaptation of \cite{shangguan2020combinatorial}'s proof of this lemma for completeness.

\begin{proof} For all $i \in \change{[t]}$, let $J_i$ be the set of coordinates for \change{which} $c_i$ and $y$ are equal. Since $c_i \in B_{\tau}(y)$, we have that $|J_i| \ge n - \tau$. Since $G$ is the generator matrix of $C$, for all $i \in \change{[t]}$, there exists a unique $f_i \in \F^k$ such that $c_i = G^{\top} f_i$. For all $i, j \in \change{[t]}$ with $i < j$ let $f_{ij} = f_j - f_i$.

  Let $v$ be a column vector of length $\binom{t}{2}k$ which is the $f_{ij}$'s concatenated together in the \change{same order as the column blocks $\{i,j\} \in \binom{[t]}{2}$ of $M$.} To \change{complete the proof, it} suffices to show that $v \neq 0$ but $Mv = 0$.

  First, to see why $v \neq 0$, note that
  \[ G^{\top}f_{01} = G^{\top} (f_1 - f_0) = c_1 - c_0 \neq 0.
  \] Thus, $f_{01} \neq 0$, so $v \neq 0$.

  Second, to see why $Mv = 0$, we split the analysis into the type (a) rows and the type (b) rows. For the type (a) rows, for each $i < j \in \change[t]$ note that $Mv$ restricted to this \change{row} block equals $f_{ij} + f_{jt} - f_{it} = (f_j - f_i) + (f_t - f_j) - (f_t - f_i) = 0$. For the type (b) rows, it suffices to check for all $i < j \in \change{[t]}$ that $G^{\top}\change{|}_{J_i \cap J_j} f_{ij} = 0$. Observe that
  \[ G^{\top}\change{|}_{J_i \cap J_j} f_{ij} = G^{\top}\change{|}_{J_i \cap J_j} (f_j - f_i) = \left.(c_j - c_i)\right|_{J_i \cap J_j}
  \] Note that, by definition, for all indices $a \in J_i \cap J_j$, we have that $(c_j)_a = y_a = (c_i)_a$. Thus, the above expression does indeed equal zero. Therefore, $Mv = 0$.

  Thus, $M$ lacks full column rank.
\end{proof}

As a result of this lemma, \cite{shangguan2020combinatorial} formulated the following conjecture whose resolution also implies that generic Reed-Solomon codes reach list-decoding capacity.

\begin{conjecture}[Conjecture 5.7 of~\cite{shangguan2020combinatorial}--restated] Let $J_1, \hdots, J_t \subseteq [n]$ be such that for all $S \subseteq [t]$,
  \begin{align} \sum_{i \in S} |J_i| - \left|\bigcup_{i \in S} J_i\right| \le (|S|-1)k,\label{eq:st20}
  \end{align} and further that (\ref{eq:st20}) is an equality when $S = [t]$. Let $G$ be a generic $\change{[n,k]}$-Vandermonde matrix. Then, $M_{G,(J_1,\hdots,J_t)}$ has full column rank.
\end{conjecture}

We now prove this conjecture.

\begin{proof} Note that $M_{G,(J_1,\hdots, J_t)}$ only includes entries from the $i$th column of $G$ if $i \in J_j$ for some $j \in [t]$. \change{In particular, $M_{G,(J_1,\hdots,J_t)}$ has full column rank if and only if $M_{G|_{J_1 \cup \cdots \cup J_t}, (J_1, \hdots, J_t)}$ has full column rank.} Thus, we may assume without loss off generality that \change{$G = G|_{J_1 \cup \cdots \cup J_t}$. In other words,}
  \begin{align} \bigcup_{i \in [t]} J_i = [n]\change{.} \label{eq:wlog}
  \end{align}

\change{Also observe that for any $i \in [t]$, if we look at (\ref{eq:st20}) with $S = [t]\setminus \{i\}$ and subtract from it the equality case of (\ref{eq:st20}) with $S = [t]$, we get that
\begin{align*}
  -|J_i| + \left|\bigcup_{j \in [t]} J_j\right| - \left|\bigcup_{j \in [t] \setminus \{i\}} J_j\right|&\le -k.
\end{align*}
Since $\left|\bigcup_{j \in [t]} J_j\right| - \left|\bigcup_{j \in [t] \setminus \{i\}} J_j\right| \ge 0$, we have that
\begin{align}
  |J_i| &\ge k. \label{eq:Jk}
\end{align}
}

  Assume for sake of contradiction that $M_{G,(J_1, \hdots, J_t)}$ lacks full column rank. Thus, there exists nonzero $v \in \F^{\binom{t}{2}k}$ such that $Mv = 0$. For each $i < j \in [n]$, let $f_{ij} \in \F^k$ be the block of $v$ corresponding to $(i,j)$. Further, define $f_t = 0$ and $f_i = -f_{it}$ for all $i \in [t-1]$. We claim that for all $i < j \in [n]$ that $f_{ij} = f_j - f_i$. This is by definition when $j = t$. Otherwise, if $j < t$, then by the type (a) rows of $M$, we may deduce that $f_{ij} - f_{it} + f_{jt} = 0$, which implies that $f_{ij} = f_j - f_i$.

  For all $i \in [t]$, define $c_i \change{= G^{\top} f_i}$. Since $v \neq 0$, we know that $f_{ij} \neq 0$ for some $i < j$. Thus, for some $i < j$, we have that $c_i \change{\neq} c_j$.

  Due to the type (b) rows of $M$, we can deduce for all $i < j \in [n]$ that $G^{\top}\change{|}_{J_i \cap J_j} f_{ij} = 0$, so $G^{\top}\change{|}_{J_i \cap J_j}f_i = G^{\top}\change{|}_{J_i \cap J_j} f_j$. Therefore, $\left.c_i\right|_{J_i \cap J_j} = \left.c_j\right|_{J_i \cap J_j}.$

  Let $y \in \F^n$ be a vector such that $y_{J_i} = \left.(c_i)\right|_{J_i}$ for all $i$. Note that at least one $y$ must exist because each of the $c_i$'s are consistent.

  For all $i \in [t]$, let $\bar{J}_i = [n] \setminus J_i$ and $u_i = c_i - y$. Note that $\supp(u_i) \subseteq \bar{J}_i$. Let $H$ be the parity-check matrix of $C$. Then, since $Hc_1 = \cdots = Hc_{t}$, we have that $Hu_1 = \cdots Hu_{t}$. Thus, the vector $(-u_1, u_2, u_3, \hdots, u_L)$ is in the kernel of the following matrix.
  \[
    \begin{pmatrix} H\change{|}_{\bar{J}_1} & H\change{|}_{\bar{J}_2} & & &\\ H\change{|}_{\bar{J}_1} & & H\change{|}_{\bar{J}_3} & &\\ \vdots & & & \ddots &\\ H\change{|}_{\bar{J}_1} & & & & H\change{|}_{\bar{J}_t}
    \end{pmatrix}
  \] Since the $c_i$'s are not all equal, the $u_i$'s are not all $0$.  Therefore, \change{by Proposition~\ref{prop:bgm}, we have that 
\begin{align*}
\sum_{i=1}^t \dim(H|_{\bar{J}_i}) - \dim(H_{\bar{J}_1} \cap \cdots \cap H_{\bar{J}_t}) \le \sum_{i=1}^t |\bar{J}_i| - 1.
\end{align*}
By (\ref{eq:Jk}), so $|\bar{J}_i| \le n-k$ for all $i \in [t]$. Since $H$ is MDS, we therefore have that $\dim(H|_{\bar{J}_i}) = |\bar{J}_i|$ for all $i \in [t]$. Thus, we can deduce that $H_{\bar{J}_1} \cap \cdots \cap H_{\bar{J}_t} \neq 0$.
}

  Observe that by (\ref{eq:st20}) we have that for any nonempty $S \subseteq [t]$, we have that
  \begin{align*} \left|\bigcap_{i \in S} \bar{J}_i\right| &= n - \left|\bigcup_{i \in S} J_i\right|\\ &\le n + (|S| - 1)k - \sum_{i \in S} |J_i|\\ &= \sum_{i \in S} |\bar{J}_i| - (|S|-1)(n - k),
  \end{align*} with equality when $S = [t]$. Thus, for any partition $P_1 \sqcup P_2 \sqcup \cdots \sqcup P_s = [t]$, we have that
  \begin{align*} \sum_{i=1}^s \left|\bigcap_{j \in P_i}\bar{J}_{j}\right| &\le \sum_{i \in [t]} |\bar{J}_i| - \sum_{i=1}^s (|P_i|-1)(n -k)\\ &= \left|\bigcap_{i \in [t]}\bar{J}_{i}\right| + (t-1)(n-k) - (t - s)(n-k)\\ &= n - \left|\bigcup_{i \in [t]}J_i\right| + (s-1)(n-k)\\ &= (s-1)(n-k),
  \end{align*} The last equality follows from our WLOG assumption (\ref{eq:wlog}). Therefore, by Theorem~\ref{thm:main-dim} the generic intersection dimension of $\bar{J}_1, \hdots, \bar{J}_t$ is zero. By Proposition~\ref{prop:genericRS_duality} and Corollary~\ref{cor:main-RS}, the dual of a generic Reed-Solomon code is $\MDS(t)$. Thus, $H_{\bar{J}_1} \cap \cdots \cap H_{\bar{J}_t} = 0$, a contradiction.
\end{proof}

\section{Computing Generic Intersection Dimension in Polynomial Time}\label{app:compute}

In this section, we present a polynomial-time algorithm for computing the generic intersection dimension of a family of sets. Let $A_1, \hdots, A_\ell \subseteq [n]$ be sets of size at most $k$. Consider the following linear program:

\begin{center} \textbf{Primal LP}
\end{center}
\begin{align*} \textbf{minimize: } & k - \sum_{i =1}^\ell \delta_i\\ \textbf{subject to: } & \forall i \in [\ell],\ \ \ & \delta_i &\ge 0\\ & \forall I \subseteq [\ell]\ (I \neq \emptyset),\ \ \ & \sum_{i \in I} \delta_i &\le k - |A_I|.
\end{align*}

By Theorem~\ref{thm:main-dim} and Lemma~\ref{lem:hall-part}, we know that the optimal \emph{integral} solution to this linear program is equal to the generic intersection dimension of $A_1, \hdots, A_\ell$. In fact, we shall prove that the optimal ``fractional'' has the same objective value.

\begin{lemma}\label{lem:duality} The objective value of the Primal LP is equal to the $k$-dimensional generic intersection dimension of $A_1, \hdots, A_\ell$.
\end{lemma}

Even with this observation, it is not obvious that the Primal LP can be solved in $\poly(\change{n,\ell})$ time, as there are roughly $2^\ell$ constraints. However, we shall demonstrate the Primal LP can still be solved efficiently.

\begin{lemma}\label{lem:efficient-LP} One can solve the Primal LP in $\poly(\change{n,\ell})$ time.
\end{lemma}

As a result of these two lemmas, we \change{prove Theorem~\ref{thm:poly-time}} as a corollary.

\change{\polytime*}

We note that while this is the same result as that proved in \change{Appendix}~\ref{sec:inv-theory}, the proof methods are very different and seem to highlight different structural aspects of generic intersections.

The remainder of this appendix is devoted to proving the two lemmas.

\subsection{Proof of Lemma~\ref{lem:duality}}

By the theory of LP duality, the objective value of the Primal LP is equal to the objective value of the Dual LP.

\begin{center} \textbf{Dual LP}
\end{center}
\begin{align*} \textbf{maximize: } & k - \sum_{\substack{I \subseteq [\ell]\\I \neq \emptyset}} (k - |A_I|) \mu_I \\ \textbf{subject to: } & \forall I \subseteq [\ell]\ (I \neq \emptyset),\ \ \ & \mu_I &\ge 0\\ & \forall i \in [\ell],\ \ \ & \sum_{\substack{I \subseteq [\ell]\\i \in I}} \mu_I &\ge 1.
\end{align*}

By Theorem~\ref{thm:main-dim} and Lemma~\ref{lem:hall-part}, we have that there exists a partition $P_1 \sqcup \cdots \sqcup P_s = [\ell]$ such that
\[ \sum_{j=1}^s |A_{P_j}| = (s-1)k + d,
\] where $d$ is the generic intersection dimension. Consider the following assignment to the Dual LP: $\mu_{P_j} = 1$ for all $j \in [s]$ and $\mu_S = 0$ otherwise. Note that since each $i \in [\ell]$ is a member of (at least) one of the $P_j$'s, we have that this assignment is feasible for the Dual LP. The objective value is then
\[ k - \sum_{j=1}^s (k - |A_{P_j}|) = \sum_{j=1}^s |A_{P_j}| - (s-1)k = d.
\]

Thus, by duality, the objective value of the Primal LP is at least the generic intersection dimension $d$. Since we previously mentioned that this value $d$ is attainable by an integral assignment (via Lemma~\ref{lem:hall-part}), we have that the objective value of the Primal LP (and the Dual LP) is exactly the generic intersection dimension.

One can also prove that the Primal LP and the Dual LP are both integral directly (i.e., without using Lemma~\ref{lem:hall-part}) using the theory of Total Dual Integrality~\cite{schrijver2003combinatorial}.

\subsection{Proof of Lemma~\ref{lem:efficient-LP}}

To solve the Primal LP, it suffices to implement an efficient (i.e., in $\poly(\change{n,\ell})$ time) separation oracle (c.f., \cite{grotschel1993geometric}). In particular, given nonnegative $\delta_1, \hdots, \delta_\ell \in \Q$, we need to efficiently compute that either (1) the $\delta$'s satisfy the primal LP or (2) exhibit a nonempty $I \subseteq [\ell]$ for which $\sum_{i \in I}\delta_i > k - |A_I|$.

Consider the function $f: 2^{[\ell]} \to \Q$ defined by
\[ f(I) = k - |A_I| - \sum_{i \in I} \delta_i.
\] with $f(\emptyset) = 0$. Observe that computing the separation oracle is thus equivalent to either verifying that $f$ is nonnegative or exhibiting an $I \subseteq [\ell]$ such that $f(I) < 0$.

It is straightforward to verify that $f$ is \emph{submodular}, that is for all $I, J \subseteq [\ell]$,
\[ f(I) + f(J) \ge f(I \cup J) + f(I \cap J).
\] Minimizing such a function can be done in $\poly(\ell)$ queries to $f$ (c.f., \cite{grotschel1993geometric}), and thus we can efficiently determine whether there exist $I \subseteq [\ell]$ such that $f(I) < 0$ (note that such an $I$ must be nonempty). Therefore, the Primal LP can be solved efficiently.

\appendix

\change{
\section{Connections to Invariant Theory}\label{sec:inv-theory}
}

The combinatorial characterization of the generic intersection dimension in Theorem~\ref{thm:main-dim} allows for a surprising connection between $\change{\MDS}(\ell)$ codes and invariant theory. A simple observation characterizes the generic intersection dimension as an instance of the well-known \change{Edmonds'} problem, i.e., a computation of the rank of a symbolic matrix with linear entries, see \cite[Appendix B]{bgm2021mds}. \change{Edmonds'} problem is not known to have a polynomial-time algorithm. Nevertheless, the non-commutative version\footnote{The symbolic variables are considered non-commutative, the base field remains commutative.} of the Edmonds' problem does have a polynomial-time algorithm, but
this is highly non-trivial and rests on some deep results in invariant
theory, see \cite{GGOW, IQS18}. The combinatorial characterization in
Theorem 1.15 allows us to prove in a curious way that the Edmonds'
problem for the generic intersection dimension is equivalent to its non-commutative counterpart, which then immediately yields a (rather distinct) polynomial-time algorithm for computing generic intersection dimension.
  
\subsection{Linear Matrices, Non-commutative Rank and the Blow-up Regularity Lemma} A matrix $L = t_1X_1 + \dots + t_mX_m$, where $X_i$ are $p \times q$ matrices with entries in a ground field $\F$ and $t_i$ are indeterminates. is called a linear matrix. There are two important notions of ranks associated to such a linear matrix are the commutative rank and the non-commutative rank. We first state their definitions and then clarify the terminology.

\begin{definition} Let $L = t_1X_1 + \dots + t_mX_m$ be a linear matrix.
\begin{itemize}
\item The commutative rank $\crk(L)$ is defined as $\rk (L)$ viewed as a matrix with entries in the function field $\F(t_1,\dots,t_m)$.
\item The non-commutative rank $\ncrk(L)$ is defined as $\rk (L)$ viewed as a matrix with entries in the free skew-field $\F\llangle t_1,\dots,t_m \rrangle$. 
\end{itemize}
\end{definition}

The reader not interested in skew-fields can very much ignore the skew-fields as long as they accept the characterization of non-commutative rank in terms of blow-ups in (\ref{eq:characterization-ncrk}) below, and perhaps take that to be the definition of non-commutative rank. For the interested readers, we give a few details, but point to references for more details.

First, we note that most of linear algebra works over skew-fields (a.k.a. division algebras). Rank is defined in terms of maximum number of (left)-linearly independent columns. This is sometimes called left column rank. Similarly one can define left and right row and column ranks. Left column rank equals the right row rank and the right column rank equals the left row rank, but the left column rank may not equal the right column rank when working over a general division algebra. For linear matrices interpreted as matrices with entries in the free skew-field, all ranks, i.e., left/right row/column ranks, are all the same.

The free skew-field itself is a technically challenging object to explain, but we will try to give a brief idea of its purpose. If you start with a (commutative) field $\F$ and add $m$ elements $t_1,\dots,t_m$ that have no extra relations on them other than the ones imposed by the axioms of a commutative field (that is one way to define a set of indeterminates), the field generated will be the function field $\F(t_1,\dots,t_m)$. The analogous object where you impose no extra relations other than the ones imposed by the axioms of a skew-field will create the free skew-field. However, unlike the well understood function field, the free skew-field is far more difficult to understand and there are several intricacies in constructing it or even showing its existence. For the reader interested in the details of free skew-fields, we refer to \cite{GGOW} for a gentle introduction and references therein for more technical details, we will not really need them here.

Computing the commutative rank of a linear matrix is \change{Edmonds'} problem and computing the non-commutative rank is the non-commutative version of \change{Edmonds'} problem. Interpolating between these two ranks are the ranks of \emph{blow-ups}, a tool that originated in invariant theory and is crucial in understanding the non-commutative rank. These blow-ups are also crucial for our purposes.

\begin{definition} \label{Def:blow-ups} Let $L = t_1X_1 + \dots + t_mX_m$ be a linear matrix of size $p \times q$. Then, for $d \in \Z_{\geq 1}$, we define the $d^{th}$ blow-up of $L$ to be the $dp \times dq$ matrix
$$
L(T_1,\dots,T_m) := X_1 \otimes T_1 + \dots + X_m \otimes T_m,
$$
where $T_i$ are $d \times d$ matrices whose entries are all distinct variables, say $t^{(i)}_{j,k}$. We define the $d^{th}$ blow-up rank of $L$ by
$$
\rk_d(L) := \crk(L(T_1,\dots,T_m)),
$$
i.e., the latter rank is taken by viewing $L(T_1,\dots,T_m)$ as a matrix with entries in the function field $\F(t^{(i)}_{j,k})$.
\end{definition}

Let us first justify the notation $L(T_1,\dots,T_m) := X_1 \otimes T_1 + \dots + X_m \otimes T_m$. To get $\sum_i X_i \otimes T_i$ from $L$ we replace each entry with a $d \times d$ matrix as follows: if the $(\alpha,\beta)^{th}$ entry of $L$ is of the form $\sum_i c_i t_i$, then we replace it with $\sum_i c_i T_i$. Thus, it is as if we plugged in $T_i$ for $t_i$. We have the equality (see \cite{IQS17,IQS18}):

\begin{equation} \label{eq:characterization-ncrk} \ncrk(L) = \lim_{d \rightarrow \infty} \frac{\rk_d(L)}{d} = {\rm sup}_d \frac{\rk_d(L)}{d} = \max_d \frac{\rk_d(L)}{d}.
\end{equation}

Dividing the rank of the blow-up by $d$ is a normalization is to be expected because we blow-up the size of the matrix by a factor of $d$. There are a few subtleties, for example the sequence $\frac{\rk_d(L)}{d}$ is not always monotone. But perhaps the major subtlety the reader may have noticed is that while $\ncrk(L)$ is an integer by definition, it is not so clear why any of the other expressions are integers. This is a consequence of an amazing result called the blow-up regularity lemma\footnote{This is just called the regularity lemma, but we call it the blow-up regularity lemma to avoid confusion with the regularity criterion.} that was first proved by Ivanyos, Qiao and Subrahmanyam in \cite{IQS17}. A more conceptual proof can be found in \cite{DM-ncrk}.

\begin{theorem} [Blow-up regularity Lemma, \cite{IQS17}] \label{Reg:IQS} Let $L = \sum_{i=1}^m t_i X_i$ be a linear matrix. Then $\frac{\rk_d(L)}{d}$ is an integer.
\end{theorem}

\subsection{Polynomial Time Computability of Generic Intersection Ranks} Let $W = (w_{ij})$ be a $k \times n$ matrix of indeterminates. Let $\mathcal{A} = (A_1,\dots,A_\ell)$ be a $\ell$-tuple of subsets of $[n]$. \change{Define}
\[ L_{\mathcal{A}}(W) = \begin{pmatrix} \change{W|_{A_1}} & \change{W|_{A_2}} & 0 & 0 & 0 \\ \change{W|_{A_1}} & 0 & \change{W|_{A_3}} & 0 & 0 \\ \vdots & 0 & 0 & \ddots & 0 \\ \change{W|_{A_1}} & 0 & 0 & 0 & \change{W|_{A_\ell}}\change{.}
\end{pmatrix}
\]

Observe that $L_{\mathcal{A}}(W)$ is a linear matrix, i.e., each entry is a linear function in the $w_{ij}$'s (it's much more special than that of course, indeed each entry is either a variable or $0$). So, we can write $L_{\mathcal{A}}(W) = \sum_{i,j} w_{ij} X_{ij}$.

The following definition is perhaps a little unnecessary, but it allows us to present easier the arguments in this section.

\begin{definition} [Generic Intersection rank] The function ${\rank_{GI}}$ takes as input the configuration $(k,n,\mathcal{A} = (A_1,\dots,A_\ell))$ where $A_i \subseteq [n]$ and returns $\dim (\bigcap_i W_{A_i})$ where $W$ is a $d \times n$ matrix of indeterminates, i.e.,
$$
\rank_{GI}(k,n,\mathcal{A}) = \dim\left( \bigcap_i W_{A_i}\right).
$$
\end{definition}

\change{By Proposition~\ref{prop:bgm}}, we have that
\begin{equation}\label{eq:gi-rank} \rank_{GI}(k,n,\mathcal{A}) = \sum_i |A_i| - \crk(L_\mathcal{A}(W)).
\end{equation}
\subsubsection{A \change{D}oubling \change{O}peration}

Given a configuration $(k,n,\mathcal{A})$, we will define a doubling operation as follows. We will assume without loss of generality that $|A_i| \leq d$ for all $i$. First, identify $[2n]$ with $\{1,\widetilde{1},2,\widetilde{2},\dots,n,\widetilde{n}\}$. Then, for each $A_i \in \mathcal{A}$, define
$$
A_i^{(2)} := A_i \cup \widetilde{A_i} \subseteq \{1,\widetilde{1},2,\widetilde{2},\dots,n,\widetilde{n}\},
$$
where $\widetilde{A_i} :=\{\widetilde{a}:a \in A_i\}$. We insist that we will order the elements in $A_i^{(2)}$ in the increasing order where the order is given by $1 < \widetilde{1} < 2 < \dots < n < \widetilde{n}$, so for example if $A_i = \{1,2\}$, then $A_i^{(2)} = \{1,\widetilde{1}, 2, \widetilde{2}\}$. Finally, define
$$
\mathcal{A}^{(2)} := (A_1^{(2)},\dots, A_\ell^{(2)}).
$$

\begin{definition} [Doubling configuration] To the configuration $\left(k,n,\mathcal{A} = (A_1,\dots,A_\ell)\right)$ with each $A_i \subseteq [n]$, we will define the doubled configuration $\left(2k,2n, \mathcal{A}^{(2)} = (A_1^{(2)},\dots, A_\ell^{(2)})\right)$ as defined above.
\end{definition}

Now, consider a $2k \times 2n$ matrix $U$ consisting of indeterminates, but let us index the rows and columns a little differently. Let us index the rows by $R = \{1,1',2,2',\dots,d,d'\}$, and the columns by $C = \{1,\widetilde{1},2,\widetilde{2},\dots,n,\widetilde{n}\}$. So, the $(i,j)^{th}$ entry of $U$ is $u_{ij}$ for $i \in R, j \in C$. So, for example if I am looking at the $(3', \widetilde{4})$ entry, I will denote it $u_{3',\widetilde{4}}$.

If one looks at the picture above of $L_\mathcal{A}(W)$ and replaces each $w_{ij}$ with $U_{ij} = \begin{pmatrix} u_{ij} & u_{i,\widetilde{j}} \\ u_{i'j} & u_{i'\widetilde{j}} \end{pmatrix}$, it should be evident that one obtains $L_{\mathcal{A}^{(2)}}(U)$ -- we picked our indexing precisely to orchestrate this. In other words, we have

$$
L_{\mathcal{A}^{(2)}}(U) = L_{\mathcal{A}}(W)(U_{11},U_{12},\dots,U_{nn}),
$$
 
where by $L_{\mathcal{A}}(W)(U_{11},U_{12},\dots,U_{nn})$, we mean take the matrix $L_{\mathcal{A}}(W)$ and replace $w_{ij}$ by $U_{ij}$ for all $i,j$ (note that this is consistent with the notation in Definition~\ref{Def:blow-ups}). In particular, since the $U_{ij}$ are $2 \times 2$ matrices whose entries are all independent indeterminates, we conclude that $\rk_2(L_{\mathcal{A}}(W)) = \crk(L_{\mathcal{A}^{(2)}}(U))$. This yields

\begin{corollary} \label{cor:doubledrank} Let $W$ be a $d \times n$ matrix with indeterminates. Let $(k,n,\mathcal{A})$ be a configuration, and let $(2k,2n,\mathcal{A}^{(2)})$ be the doubled configuration. Then
$$
\rank_{GI}(2k,2n,\mathcal{A}^{(2)}) = 2 (\sum_i |A_i|) - \rank_2(L_\mathcal{A}(W)).
$$
Further, this means that $\rank_{GI}(2k,2n,\mathcal{A}^{(2)})$ is a multiple of $2$.
\end{corollary}

\begin{proof} The equality $\rank_{GI}(2k,2n,\mathcal{A}^{(2)}) = 2 \sum_i |A_i| - \rank_2(L_\mathcal{A}(W))$ follows from the above discussion along with (\ref{eq:gi-rank}). By the blow-up regularity lemma, i.e., Theorem~\ref{Reg:IQS}, we know that $\rank_2(L_\mathcal{A}(V))$ is a multiple of $2$, and hence so is $\rank_{GI}(2k,2n,\mathcal{A}^{(2)})$.
\end{proof}

Analogously, for $t \in \Z_{\geq 1}$, we define the $t-$pled configuration $\left(tk,tn, \mathcal{A}^{(t)} = (A_1^{(t)},\dots, A_\ell^{(t)})\right)$.

\begin{corollary} \label{cor:t-pledrank} Let $W$ be a $k \times n$ matrix with indeterminates. Let $(k,n,\mathcal{A})$ be a configuration, and let $(tk,tn,\mathcal{A}^{(t)})$ be the $t-$pled configuration. Then
$$
\rank_{GI}(tk,tn,\mathcal{A}^{(t)}) = t (\sum_i |A_i|) - \rank_t(L_\mathcal{A}(W)).
$$
Further, this means that $\rank_{GI}(tk,tn,\mathcal{A}^{(t)})$ is a multiple of $t$.
\end{corollary}

\subsubsection{Scalability of \change{G}eneric \change{I}ntersection \change{D}imension}

\begin{lemma} \label{lem:gi-rank-scales} Let $(k,n,\mathcal{A})$ be a configuration, and let $(tk,tn,\mathcal{A}^{(t)})$ be the $t-$pled configuration. Then
$$
\rank_{GI}(tk,tn,\mathcal{A}^{(t)}) = t \cdot \rank_{GI}(k,n,\mathcal{A}).
$$
\end{lemma}

\begin{proof} Suppose $\mathcal{A}$ consists of $\ell$ sets. Then, so does $\mathcal{A}^{(t)}$. With this observation, the lemma follows immediately from Theorem~\ref{thm:main-dim} since the dimension of the $t-$pled configuration $(tk,tn,\mathcal{A}^{(t)})$ is also a maximum over all partitions of $\ell$ (i.e., the number of parts in $\mathcal{A}$) and each corresponding term is $t$ times larger than the one for the configuration $(k,n,\mathcal{A})$.
\end{proof}

\begin{corollary} Given the configuration $(k,n,\mathcal{A})$, we have
$$
\rank_{GI}(k,n,\mathcal{A}) = \sum_i |A_i| - \ncrk(L_\mathcal{A}(W)).
$$
\end{corollary}

\begin{proof} From the above lemma and Corollary~\ref{cor:t-pledrank}, it follows that for any $t \in \Z_{\geq 1}$,
$$
t (\sum_i |A_i|) - \rank_t(L_\mathcal{A}(W)) =\rank_{GI}(tk,tn,\mathcal{A}^{(t)}) = t (\rank_{GI}(k,n,\mathcal{A})) = t (\sum_i |A_i|) - t \cdot \crk (L_\mathcal{A}(W)).
$$
This means that $t \cdot \crk(L_\mathcal{A}(W)) = \rank_t(L_\mathcal{A}(W))$, so
$$\ncrk(L_\mathcal{A}(W)) = \lim_{t \to \infty} \rank_t(L_\mathcal{A}(W))
/t = \crk(L_\mathcal{A}(W)).$$ Thus, we conclude that
\[ \rank_{GI}(k,n,\mathcal{A}) = \sum_i |A_i| - \ncrk(L_\mathcal{A}(W)).\qedhere
\]
\end{proof}

\begin{proof}[Proof of Theorem~\ref{thm:poly-time}] Let $\mathcal{A} = (A_1,\dots,A_\ell)$. Then, in the notation of this section, we have $\dim(W_{A_1} \cap W_{A_2} \cap \dots W_{A_\ell}) = \rank_{GI}(k,n,\mathcal{A})$. By the above corollary, we have $\rank_{GI}(k,n,\mathcal{A}) = \sum_i |A_i| - \ncrk(L_\mathcal{A}(W))$, so it suffices to compute $\ncrk(L_\mathcal{A}(W))$. But there is a polynomial time algorithm for this \cite{IQS18}.\footnote{Alternately, one can use the algorithm in \cite{GGOW}, but this needs to be appropriately modified because the algorithm as stated can only check if the non-commutative rank of a square matrix is full or not.} Note that the size of $L_\mathcal{A}(W)$ is $\change{\poly(\change{n,\ell})}$\change{.}
\end{proof}

\bibliographystyle{alpha}
\bibliography{references}

\end{document}